\documentclass[11pt]{article}
\usepackage{todonotes}
\usepackage{float}
\usepackage{amsmath}
\usepackage{amsthm}
\usepackage{amsfonts}

\usepackage[top=7cm,bottom=2.5cm,headheight=120pt,headsep=15pt,left=6cm,right=1.5cm,marginparwidth=4.2cm,marginparsep=0.5cm]{geometry}

\usepackage{marginnote}
\reversemarginpar  
\usepackage{lipsum} 
\usepackage{calc}
\usepackage{siunitx}
\usepackage{lineno}
\usepackage[utf8]{inputenc}
\usepackage[T1]{fontenc}
\usepackage[
backend=biber,
natbib=true,
sortcites=true,
defernumbers=true,
style=authoryear,
citestyle=authoryear-comp,
maxnames=99,
maxcitenames=2,
giveninits=true,
terseinits=true,
url=false
]{biblatex}
\DeclareNameAlias{default}{family-given}
\renewbibmacro{in:}{%
  \ifentrytype{article}{}{\printtext{\bibstring{in}\intitlepunct}}} 
\DeclareFieldFormat[article]{pages}{#1} 
\AtEveryBibitem{\ifentrytype{article}{\clearfield{number}}{}} 
\DeclareFieldFormat[article, inbook, incollection, inproceedings, misc, thesis, unpublished]{title}{#1} 
\usepackage{csquotes}
\RequirePackage[english]{babel} 
\DeclareBibliographyCategory{ignore}
\addtocategory{ignore}{recommendation} 
\usepackage{nameref}
\usepackage[pdfborder={0 0 0}]{hyperref}  
\usepackage{graphbox}  
\usepackage{microtype}
\DisableLigatures[f]{encoding = *, family = * }
\RequirePackage[default,scale=0.90]{opensans}
\usepackage{xcolor}
\definecolor{darkgray}{HTML}{808080}
\definecolor{mediumgray}{HTML}{6D6E70}
\definecolor{ligthgray}{HTML}{d9d9d9}
\definecolor{pciblue}{HTML}{74adca}
\definecolor{opengreen}{HTML}{77933c}
\usepackage{changepage}
\usepackage[aboveskip=1pt,labelfont=bf,labelsep=period,singlelinecheck=off]{caption}

\usepackage{fancyhdr}  
\usepackage{lastpage}  
\fancyhfoffset[L]{4.5cm}  
\renewcommand{\headrulewidth}{\ifnum\thepage=1 0.5pt \else 0pt \fi} 

\newcommand{\DOIrecommendationlink}{\href{https://doi.org/\DOIrecommendation}{https://doi.org/\DOIrecommendation}}

\newcommand{\PCI}{Peer Community In Mathematical and Computational Biology}

\newcommand{\beginingpreprint}{
\vspace*{0.5cm}
\begin{flushleft}
\baselineskip=30pt
\marginpar{
\large\textnormal{\color{pciblue}\\RESEARCH ARTICLE}\\
\vspace*{0.5pt}
\\
\includegraphics[align=c,width=0.5cm]{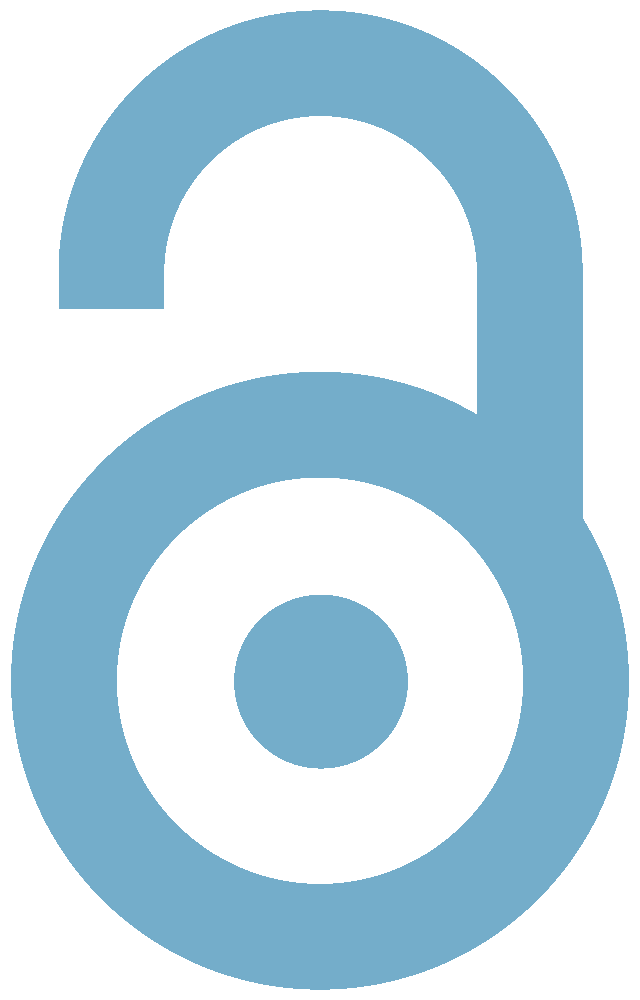} \space \large\textbf{\color{pciblue}Open Access}\\
\\
\includegraphics[align=c,width=0.5cm]{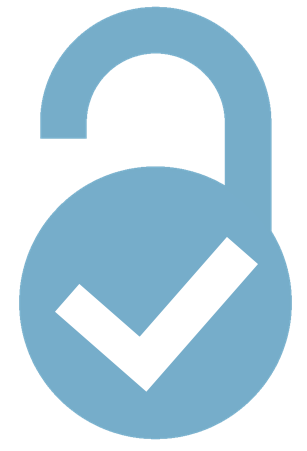} \space \large\textbf{\color{pciblue}Open Peer-Review}\\
\\
\\
\\
\\
\\
\\
\raggedright
\scriptsize\textbf{Cite as:}\space
\citeas\\
\vspace*{0.5cm}
\textbf{Posted:} \datepub\\
\vspace*{0.5cm}
\textbf{Recommender:}\\
\recommender\\
\vspace*{0.5cm}
\textbf{Reviewers:}\\
\reviewers\\
\vspace*{0.5cm}
\textbf{Correspondence:}\\
\href{mailto:\email}{\email}\\

}
{\Huge
\fontseries{sb}\selectfont{\preprinttitle}}
\end{flushleft}
\vspace*{0.25cm}
\begin{flushleft}

\Large
\listauthors
\end{flushleft}
\bigskip
{\raggedright
\listinstitutions}
\begin{flushleft}
\fcolorbox{lightgray}{lightgray}{
\parbox{\textwidth - 2\fboxsep}{
\centering\large{\fontseries{sb}\selectfont{This article has been peer-reviewed and recommended by\\
\emph{\PCI} (\DOIrecommendationlink)}}\\
}}
\end{flushleft}
\vspace*{0.5cm}
\fcolorbox{pciblue}{pciblue}{
\parbox{\textwidth - 2\fboxsep}{
\vspace{0.25cm}
\textbf{\large{\textsc{Abstract}}}\\
\preprintabstract\\

\footnotesize{\textbf{\emph{Keywords: }}\preprintkeywords}
\vspace{0.25cm}}
}
}

\usepackage[ruled,linesnumbered,noend]{algorithm2e}
\usepackage{hyperref}
\usepackage{color}
\usepackage{url}
\usepackage{xspace}

\usepackage{graphics}

\usepackage{epstopdf}
\usepackage{caption}
\usepackage{subcaption}

\newtheorem{nclaim}{Claim}

\newcommand{\isep}{\mathrel{{.}\,{.}}\nobreak}

\newcommand{\lca}{\textsc{lca}}

\renewcommand{\l}{\ell}

\newcommand{\genes}{\Gamma}
\newcommand{\species}{\Sigma}

\newcommand{\ml}[1]{\textcolor{orange}{#1}}
\newcommand{\cs}[1]{\textcolor{magenta}{#1}}
\newcommand{\man}[1]{\textcolor{black}{#1}}

\newcommand{\mj}[1]{#1}
\newcommand{\mjn}[1]{\textcolor{blue}{#1}}

\renewcommand{\mjn}[1]{\textcolor{blue}{#1}}
\renewcommand{\cs}[1]{\textcolor{blue}{#1}}
\renewcommand{\ml}[1]{\textcolor{blue}{#1}}

\renewcommand{\mjn}[1]{#1}
\renewcommand{\cs}[1]{#1}
\renewcommand{\ml}[1]{#1}

\usepackage[normalem]{ulem}

\newcommand{\sdon}[3]{don^{#3}_{{#1}\rightarrow{#2}}}
\newcommand{\srec}[3]{rec^{#3}_{{#1}\leftarrow{#2}}}

\makeatletter
\def\thm@space@setup{%
  \thm@preskip=3pt \thm@postskip=3pt
}
\makeatother

\newtheorem{definition}{Definition}
\newtheorem{lemma}{Lemma}
\newtheorem{theorem}{Theorem}

 \newcommand{\D}{\ensuremath{\mathbb{D}}\xspace}

\newcommand{\DTL}{\ensuremath{\mathbb{DTL}}\xspace}
\newcommand{\T}{\ensuremath{\mathbb{T}}\xspace}
\newcommand{\TL}{\ensuremath{\mathbb{TL}}\xspace}
\newcommand{\TS}{\ensuremath{\mathbb{T^S}}\xspace}
\newcommand{\TD}{\ensuremath{\mathbb{T^D}}\xspace}

\renewcommand{\S}{\ensuremath{\mathbb{S}}\xspace}
\newcommand{\SL}{\ensuremath{\mathbb{SL}}\xspace}

\newcommand{\last}{\textsc{last}}

\begin{document}

\newcommand{\preprinttitle}{Consistency of orthology and paralogy constraints in the presence of gene transfers}

\newcommand{\listauthors}{\raggedright 
Mark Jones\textsuperscript{1}, \space
Manuel Lafond\textsuperscript{2}, 
Celine Scornavacca\textsuperscript{3}
}

\newcommand{\listinstitutions}{
\textsuperscript{1} Delft Institute of Applied Mathematics, Delft University of Technology, The Netherlands, \texttt{M.E.L.Jones@tudelft.nl}
\\
\textsuperscript{2} Departement d'informatique, Université de Sherbrooke, Canada, \texttt{manuel.lafond@USherbrooke.ca}
\\
\textsuperscript{3} Institut des Sciences de l'Evolution, 
Universit\'{e} de Montpellier, CNRS, IRD, EPHE
34095 Montpellier Cedex 5 - France 
\texttt{Celine.Scornavacca@umontpellier.fr}
}

\newcommand{\datepub}{15 February 2022}

\newcommand{\recommender}{Barbara Holland}

\newcommand{\DOIrecommendation}{10.24072/pci.mcb.100009}

\newcommand{\citeas}{Jones M, Lafond M, Scornavacca C (2022) Consistency of orthology and paralogy constraints in the presence of gene transfers. arXiv:1705.01240 [cs], ver.6 peer-reviewed and recommended by Peer Community in Mathematical and Computational Biology. https://arxiv.org/abs/1705.01240.}

\newcommand{\preprintkeywords}{Algorithms, phylogenetics, orthology, horizontal gene transfer}

\newcommand{\email}{M.E.L.Jones@tudelft.nl, manuel.lafond@USherbrooke.ca, and \\ Celine.Scornavacca@umontpellier.fr}

\newcommand{\reviewers}{Two anonymous reviewers.}



\newcommand{\preprintabstract}{ 
Orthology and paralogy relations are often inferred by methods based on gene sequence similarity \cs{that} yield 
a graph depicting the relationships between gene pairs.  Such relation graphs frequently contain errors, 
as they cannot be explained via a gene tree that contains the depicted orthologs/paralogs while being consistent with the species evolution.  
Previous research has mostly focused on correcting such errors in some minimal way, for instance by changing a minimum number of relations to attain consistency.

In this work, we ask: could the errors in the orthology predictions be explained by lateral gene transfer?  We formalize this question
by allowing gene transfers to behave either as a speciation or as a duplication, expanding the space of valid orthology graphs.
We then provide a variety of algorithmic
results regarding the underlying problems.  Namely, we show that deciding if a relation graph $R$ is consistent with a given species network $N$ with known transfer highways is 
NP-hard, and that it is W[1]-hard under the parameter ``minimum number of transfers''.  
During the process, we define a novel algorithmic problem called \emph{Antichain on trees}, which may be useful for other reductions.
We then present an FPT algorithm for the decision problem
based on the degree of the gene tree associated with $R$.
We also study analogous problems in the case that the transfer highways on a species tree are unknown.
}

\beginingpreprint

\section{Introduction}


In phylogenetics, evolutionary relationships between genes and species are often represented via phylogenetic trees.
\emph{Species trees} are phylogenetic trees displaying the evolutionary relationships among a set of species, while  \emph{gene trees} are phylogenetic trees displaying the evolutionary relationships among genes.
Vertical descent with modification (speciation)
constitutes only part of the events shaping
a gene history;
\mj{other such events include,}
for example,  duplications, 
losses and transfers of genes. 


When gene trees are used to estimate the evolutionary relationships of the species containing 
\mj{those genes,}
only \emph{homologous} genes -- genes sharing a common ancestor --  should be compared. 
Homology can be refined into the concepts of  \emph{orthology} and \emph{paralogy}: two genes  from two different species  are said to be \emph{orthologous} if  they are derived from a single gene present in the last common ancestor of the two species via a speciation event, and \emph{paralogous} if  they were derived via a duplication event \citep{fitch1970distinguishing}. 

Orthology inference is the starting point of several comparative genomics studies, and is also a key instrument for  functional annotation of new genomes~\citep{gabaldon2013functional}. Several 
methods have been designed to distinguish orthologs from paralogs.
These can be roughly divided in two groups \citep{altenhoff2012inferring}. The first group of methods, based on phylogenetic inference, 
reconstruct a \emph{gene tree}
and deduce orthology relationships from this tree by comparing it with the species tree via \emph{reconciliation algorithms} 
(see \cite{boussau:hal-02535529} 
for a review). 
Another class of methods estimates orthology using sequence similarity \ml{(see e.g.~\cite[among others]{li2003orthomcl,emms2015orthofinder} and~\cite{kristensen2011computational} for a survey)},
hypothesising that orthologs are more similar than paralogs.
Both methods can yield a \emph{relation graph}, in which vertices are genes, edges represent putative orthologous gene pairs and non-edges represent putative paralogs.
Phylogeny-based methods
require a prior knowledge of the species tree, and are very dependent on the accuracy of the gene trees. Unfortunately, the species phylogeny is not always known and gene trees  can be highly inaccurate as a result of several kinds of reconstruction artefact, e.g. long-branch attraction (LBA)  \citep{bergsten2005review}.
Similarity-based methods do not suffer from these drawbacks but still have an important weakness: 
the inferred relation graph $R$ may fail to be 
\emph{consistent}, meaning that there is no gene tree, labeled by speciation and duplication events, that can both explain the relations depicted by $R$
and ``agree'' with a known species tree $S$. 
Moreover, approaches based on sequences tend to miss orthologs whose evolutionary path involves a duplication followed by high divergence, which occurs for instance in neofunctionalisation~\citep{lafond2018accurate}.

In 
\mj{recent years,}
the decision problems of consistency of orthology/paralogy relations have been extensively studied \citep{hernandez2012event,hellmuth2013orthology,lafond2014orthology,hellmuth2015phylogenomics,jones2016,lafond2016link,dondi2017approximating}.
Two possible explanations for the inconsistency of a relation graph $R$ are that either the set of relations contains errors, 
or the evolutionary model used to assess consistency is not appropriate for the gene family at hand.
Most of the previous work in this field has been devoted to detection and correction of errors in relation graphs.  
\man{The second possibility has recently been considered in~\cite{hellmuth2019reconciling}.  The authors ask, given a \cs{event-labeled} gene tree $G$ that displays a given \cs{set} of relations, whether there is a species network $N$ that \cs{can} be reconciled with $G$.}
In a similar vein, 
in this paper we ask: 
can inconsistent relations be explained by extending the usual speciation/duplication model to lateral gene transfers?
Two genes are said to be \emph{xenologous} if at least one of the two genes has been acquired by gene transfer.  
As discussed in~\cite{koonin2005orthologs}, genes related by transfer may appear either as orthologs or paralogs, even though they are not related by speciation or duplication at their lowest common ancestor.
The terms \emph{pseudoorthologs} and \emph{pseudoparalogs} were used to designate homologous genes mimicking orthology and paralogy, respectively, 
after one or more lateral gene transfers.  Here, we provide a variety of algorithmic results regarding the question 
of explaining inconsistent relations using these new types of relations.

%
The paper is organized as follows.
In Section~\ref{sec:prelim}, we introduce the notion of orthology/paralogy consistency with a given species network $N$, 
and show how it relates to $DS$-trees, which are gene trees labeled by speciation and duplication only.
Then, in Section~\ref{sec:w1hardness} we study the question of deciding whether a relation graph $R$ is 
consistent with $N$, meaning that $R$ can be represented by a gene history, possibly undergoing lateral transfers,
that agrees with $N$.  
We show that, unfortunately, 
this is 
\mj{an NP-hard problem. Furthermore, the problem is unlikely to be fixed-parameter tractable with respect to the number of transfers, as this parameterized version of the problem is $W[1]$-hard.}
On the positive side, we show in Section~\ref{sec:dpalgo} that these problems can be solved in time
\mj{$O(2^{k}k!k|V(R)||V(N)|^4)$,}
where here 
$k$ is the maximum degree of the smallest $DS$-tree exhibiting the relations of $R$.
In Section~\ref{sec:unknownhighways}, we turn to the variant where we have a species tree $S$ rather than a network, and 
ask if transfer arcs can be inserted into $S$ so that $R$ becomes consistent.  
Some proofs are quite technical and can be found in the Appendix.

\vspace{-2mm}

\section{Preliminaries}\label{sec:prelim}


We use the notation $[n] = \{1, 2, \ldots, n\}$.
Across the paper, let $\genes$ a set of genes, $\species$ a set of species, and  
$\sigma : \genes \rightarrow \species$ the mapping between genes and species.

\begin{figure*}[t!]
    \centering
    \begin{subfigure}[t]{0.3\textwidth}
        \centering
\includegraphics[height=1.0\linewidth]{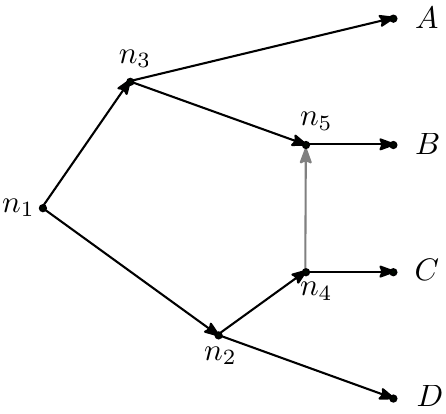}
        \caption{An LGT network $N$\label{fig1a}}
    \end{subfigure}%
    \hspace{0.9cm}
    \begin{subfigure}[t]{0.3\textwidth}
        \centering
\includegraphics[height=1.0\linewidth]{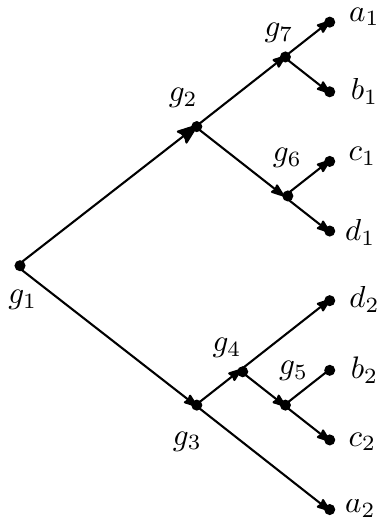}
        \caption{A gene tree $G$\label{fig1b}}
    \end{subfigure}
        \hspace{-0.5cm}
        \begin{subfigure}[t]{0.3\textwidth}
        \centering
\includegraphics[height=1.0\linewidth]{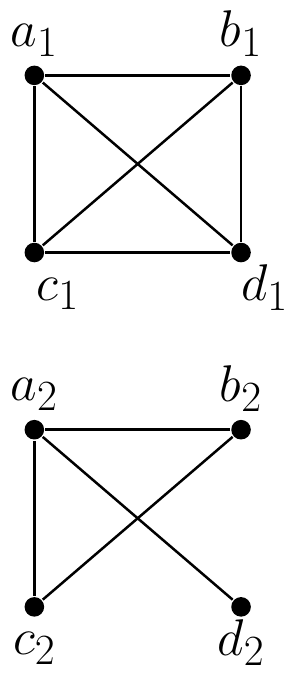}
        \caption{A relation graph $R$\label{fig1c}}
    \end{subfigure}
    \caption{An illustration of an LGT network \mj{with secondary arc $(n_4,n_5)$}, a gene tree and a relation graph. The genes $a_i$, $b_i$, $c_i$ and $d_i$, with $i \in \{1,2\}$, belong respectively to species $A$, $B$, $C$ and $D$. Internal nodes are labeled only for the purpose of giving an example of a reconciliation between $N$ and $G$, see main text. \cs{$R$ is not $T_0(N)$-consistent but it is $N$-consistent using 1 transfer.}  \label{fig1}}
\end{figure*}

All trees in this paper are assumed to be rooted and directed, each edge being oriented away from the root.
A \emph{species network} $N$ on $\species$ is a directed acyclic graph with a single indegree-0 node (the \emph{root}) and $|\species|$ outdegree-0 nodes (the \emph{leaves}), such that each leaf  is labeled by a different element of $\species$. 
Here we will consider only \emph{binary} species networks, 
\mj{in which}
internal nodes have either indegree 1 and outdegree 2 (\emph{principal} nodes) or indegree 2 and outdegree 1 (\emph{secondary} nodes or \emph{reticulations}). 
\cs{A \emph{Lateral Gene Transfer (LGT) network}} $N$  is a species network along with a partition of $E(N) = E_p \cup E_s$ into a set of \emph{principal arcs} $E_p$ and a set of  \emph{secondary arcs} $E_s$~\citep{Cardona2015}.  The $E_p$ edges correspond to vertical descent, whereas the $E_s$ edge correspond to pairs of species that may transfer genetic content.  The subnetwork $N' = (V(N), E_p)$ obtained after removing the $E_s$ edges must be a tree \ml{in which the root has outdegree $2$}.  We denote by $T_0(N)$ the tree obtained from $N'$ after \cs{suppressing} indegree-1 outdegree-1 nodes. Roughly speaking, an LGT network can also be seen as a network obtained by starting with a species tree $S = T_0(N)$, and then adding secondary arcs with endpoints located on the edges of $S$.  Note that LGT networks are \emph{tree-based networks}, where $T_0(N)$ is a \emph{distinguished base tree}  \citep{Francis2015}. 
As defined in~\cite{Gorecki2004,nojgaard2018time}, we say that an LGT network $N$ is \emph{time-consistent} if there exists a function $t:V(N) \rightarrow \mathbb{N}$ such that: 
\begin{enumerate}
\item $t(u)=t(v)$, if $(u,v) \in E_s$, and
\item $t(u)<t(v)$, if $(u,v) \in E_p$.
\end{enumerate}

\ml{Note that although time-consistency forbids directed cycles, not all directed acyclic graphs are time-consistent.  For instance, one can easily construct an acyclic LGT network that contains two principal arcs $(a, b)$ and $(c, d)$, and secondary arcs $(a, d)$ and $(b, c)$; no time-consistent labeling is possible for $a, b, c, d$.  It is also worth mentioning that LGT networks that admit a time-consistent map were characterized in~\cite{Gorecki2004}, where a linear-time algorithm is given to find such a map.}

Here a \emph{gene tree} $G$ on $\genes$ is a \cs{binary} tree with $|\genes|$ leaves such that each leaf is labeled by a different element of $\genes$.  

For a binary network $N$, the root node is denoted by $r(N)$, the set of leaves is denoted by ${L}(N)$  and the set of internal nodes is denoted by $I(N)$.  
An internal node $x$ of $N$ has either two children, which we will usually denote $x_l$ and $x_r$, or one child, which we will denote $x_l$.  The parent of a node $x$ of in-degree $1$ is denoted $p(x)$.
If $x$ has out-degree $2$, the subnetwork rooted at $x$, denoted $N_x$, is the network consisting of the root $x$ and all the nodes
reachable from $x$ (hence if $N$ is a tree, then $N_x$ is a subtree).  If $N$ is a rooted tree, $\lca(x, y)$ denotes the lowest common ancestor of $x$ and $y$.
Note that all these notations apply to LGT networks and to gene trees (which are special cases of networks).  
If $N$ is a species network, since $L(N)$ and $\species$ are in bijection, we will not make the distinction between a leaf of $N$ and a member of $\species$.  The same applies to gene tree leaves and $\genes$.

\vspace{-3mm}

\subsection{Reconciliations between gene trees and species networks}

\mj{A \DTL reconciliation aims at explaining how an evolutionary history for a family of genes (given by a gene tree) may fit within a given species network $N$, using speciation, duplication, transfer and gene loss events.
The internal nodes of gene trees, representing ancestral genes, are mapped to ancestral species.}
Furthermore, the branches of a gene tree
may hide multiple events that have not been observed, mainly due to losses.  Hence, a reconciliation $\alpha$ maps a node $x$ of $G$ 
to the sequence of species for the genes that should appear on its parent branch. 
Possible mappings are restricted by few  conditions aimed at describing only biologically-meaningful evolutionary histories. 

\ml{A reconciliation model for gene trees and time-consistent LGT networks (called H-trees) was proposed in~\cite{Gorecki2010,gorecki2012inferring}, along with algorithms to minimize the duplication, loss and transfer cost.}
\ml{We use~\cite[Definition 3]{Scornavacca2016}, which uses the following formalization:}

   \begin{definition}[\cite{Scornavacca2016}]
   \label{def:DTLrecLGTNetwork}
     Given an LGT network $N$ and a gene tree $G$, let $\alpha$ be a function that maps each node $u$ of $G$ onto a directed path of $N$, denoted $\alpha(u) = (\alpha_1(u), \ldots, \alpha_{\ell}(u))$.  Then $\alpha$ is a \emph{\DTL reconciliation} 
between $G$ and $N$ if and only if exactly one of the following events
occurs for each node $u$ of $G$ and each $\alpha_i(u)$. 
\mjn{For each $\alpha_i(u)$ we also specify a label $e_\alpha(u,i)$ corresponding to the case that holds between $u$ and $\alpha_i(u)$, given in square brackets below}
(for simplicity, let $x:=\alpha_i(u)$ below): 

  \begin{itemize}
  \item [a)] if $x$ is the last node of $\alpha(u)$, one of the cases below is true:
    \begin{enumerate}
    \item[1.]  $u \in L(G)$, $x \in L(N)$ and  $\sigma(u)=x$;  \hfill [extant leaf]
    \item[2.] $\{\alpha_1(u_l), \alpha_1(u_r)\} =\{x_l,x_r\}$, where $(x, x_l), (x, x_r) \in E_p$;\hfill $[\S]$
    \item[3.] $\alpha_1(u_l)=x$ and $\alpha_1(u_r)=x$; \hfill $[\D]$
    \item[4.] $\{\alpha_1(u_l), \alpha_1(u_r)\} = \{x, y\}$, where $(x,y) \in E_s$;  \hfill $[\T]$
    

   \end{enumerate}
  \item  [b)] otherwise, one of the cases below is true:
   \begin{enumerate}
    \item[5.]  $\alpha_{i+1}(u) =y$, where $(x,y)$ \cs{is one of the two outgoing arcs of $x$ in $E_p$}; \hfill$[\SL]$ 
     \item[6.] $\alpha_{i+1}(u)=y$, where $(x,y)$ is in $E_s$;  \hfill   $[\TL]$
     \item[7.] $\alpha_{i+1}(u)=y$ and $(x,y)$ is the only outgoing arc of $x$ in $E_p$; \hfill   $[\emptyset]$
    \end{enumerate}
  \end{itemize}
  \mj{When $\alpha$ is a \DTL reconciliation between $G$ and $N$, we call the pair $(G,\alpha)$ a \emph{reconciled gene tree}.}
   \end{definition}

   
By a slight abuse of notation, we may write $|\alpha(u)|$ to denote the number of vertices on the path $\alpha(u)$.
If $\alpha$ is clear from the context, we may write $e(u, i)$ \mjn{in place of $e_{\alpha}(u,i)$}.  
With a slight abuse of terminology, we will write $e(\alpha_i(u))$ to denote $e(u,i)$. 
We will also write \mj{$\alpha_{\last}(u)$ to denote $\alpha_\l(u)$ and} $e(u, \last)$ or $e(\alpha_{\last}(u))$ to denote $e(u, \l)$ where $\l = |\alpha(u)|$.   
   
A speciation (\S) sends its child genes to the child species through principal arcs.  A duplication (\D) makes two copies of the gene in the current species. A  transfer (\T) corresponds to transferring the lineage of a child of a gene to another branch of the species tree,  
while the sibling lineage still evolves within the lineage of the parent. A speciation-loss (\SL) is a speciation where one of the descending genes is absent.  A transfer-loss (\TL)  is a transfer of one of the two descendants of a gene combined with the loss of its sibling lineage. A no event ($\emptyset$) indicates that the gene is not transferred and follows the primary species history. 
Note that, if $N$ is time-consistent,  all \T and \TL events can be guaranteed to happen between co-existing species.
\ml{Moreover, it is not hard to see that for a given root-to-leaf path $g_1, \ldots, g_k$ of $G$, the concatenation of the $\alpha(g_i)$ paths correspond to a directed path in $N$ (with some nodes that may occur multiple times in a row because of $\mathbb{D}$ nodes).  Hence, if $N$ is time-consistent, $\alpha$ ensures that genes evolve without going back in time.  Also note that some models only specify the last element of each $\alpha(u)$ (e.g. the $\mu$ map in~\cite{lafond2020reconstruction,nojgaard2018time}).}

An example of a \DTL reconciliation between the LGT network in Figure \ref{fig1a} and the gene tree in  Figure \ref{fig1b}  is as follows: $\alpha(g_1)=(n_1)$, $\alpha(g_2)=(n_1)$, $\alpha(g_3)=(n_1)$, $\alpha(g_4)=(n_2)$, $\alpha(g_5)=(n_2,n_4)$, $\alpha(g_6)=(n_2)$, $\alpha(g_7)=(n_3)$, $\alpha(a_1)=(A)$, $\alpha(b_1)=(n_5,B)$, $\alpha(c_1)=(n_4,C)$, $\alpha(d_1)=(D)$, $\alpha(a_2)=(n_3,A)$, $\alpha(b_2)=(n_5,B)$, $\alpha(c_2)=(C)$, $\alpha(d_2)=(n_2,D)$.
\mjn{See Figure \ref{fig:embedding}}.
\mj{For this \DTL reconciliation, we have  $e(\alpha_1(g_1)) = e(\alpha_1(g_4)) = \D$, $e(\alpha_1(g_2)) = e(\alpha_1(g_3)) = e(\alpha_1(g_6)) = e(\alpha_1(g_7)) = \S$, $\cs{e(\alpha_2(g_5))} = \T$,  $e(\alpha_1(b_1)) = \cs{e(\alpha_1(b_2))} = \emptyset$, $e(\alpha_1(c_1)) = \TL$,
$e(\alpha_1(a_2)) = e(\alpha_1(d_2)) = \cs{e(\alpha_1(g_5))} = \SL$, and $e(\alpha_{\last}(u) )=$ extant leaf for all $u \in L(G)$.}

\begin{figure*}[t!]
    \centering
\includegraphics[height=0.4\linewidth]{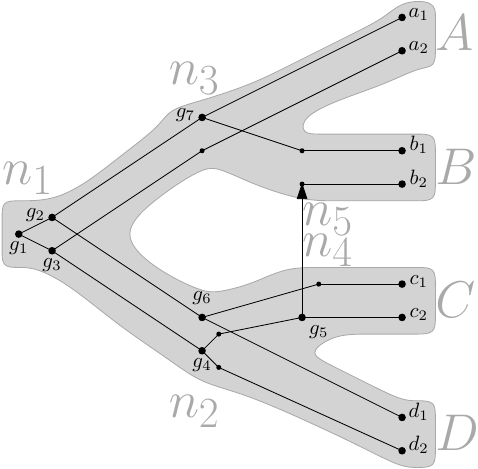}
    \caption{Illustration of a $\DTL$ reconciliation $\alpha$ between the LGT network $N$ in Figure \ref{fig1a} and the gene tree $G$ in  Figure \ref{fig1b}.
In cases where $\alpha(x)$ is a path with more than one vertex, only the last vertex of this path is labeled with $x$. Labels in grey denote the vertices of $N$.}  \label{fig:embedding}
\end{figure*}


Given $x,y \in \genes$, let  $u = \lca_G(x,y)$. 
Then we say that $x$ and $y$ are
\emph{orthologs} w.r.t a reconciled gene tree  $G$ if 
$e(\alpha_\last(u))  = \S$, 
\emph{paralogs} if  
$e(\alpha_\last(u))  = \D$, 
and \emph{xenologs}  if 
$e(\alpha_\last(u))  = \T$.
Note that 
one of these cases must hold for all distinct 
$x,y \in \genes$.

\vspace{-3mm}

\subsection{Orthology/paralogy relation graphs}

\ml{An undirected} graph $R$ is called a \emph{relation graph} if $V(R)=\genes$ (see Figure \ref{fig1c}).  \ml{Since $R$ is undirected, we may denote an edge $\{x, y\}$ of $R$ as $xy$.}
Relation graphs are often used to depict orthology and paralogy relationships \citep{hellmuth2013orthology}:  for any pair $x,y$ of \ml{distinct} vertices in $R$, $xy$ is an edge in $R$ if  $x$ and $y$ are orthologs, otherwise  $x$ and $y$ are paralogs.
Several orthology-detection methods such as OrthoMCL \citep{li2003orthomcl}, ProteinOrtho~\citep{lechner2011proteinortho} and OrthoFinder \citep{emms2015orthofinder} use sequence similarity as a proxy for orthology.  Roughly speaking, similar sequences are presumed more likely to be orthologs. 
When transfers are present, \ml{sequence similarity predictions} get trickier: xenologs can be ``interpreted'' as either orthologs, in case the two copies retained the same function (and thus their sequences are likely to be similar), or paralogs, if they did not (and thus their sequences are likely to be different).  In the following, we adapt the framework described in  \cite{hellmuth2013orthology} to the presence of xenologs. Note that in~\cite{hellmuth2016mathematics,geiss2017short,geiss2018reconstructing,lafond2020reconstruction}
, the authors approach this problem from a different angle, supposing the xenology relationships are given in the relation graph. 

We say that a reconciled gene tree $(G, \alpha)$ \emph{displays} a  relation graph $R$,
if there is a way of reinterpreting transfers as either speciation or duplication events,
such that  for any pair $x,y$ of vertices in $R$,
$xy$ is an edge in $R$  if and only if $x$ and $y$
are orthologs according to $(G, \alpha)$. More precisely, we introduce two new types of events $\TS, \TD$, which correspond to transfers that behave as a speciation and a duplication, respectively.  We then have the following definition:

\begin{definition}
 Let $N$ be an LGT network, $R = (\genes, E)$ a relation graph, and $(G, \alpha)$ a reconciled gene tree with respect to $N$.
 We say that $(G, \alpha)$ \emph{displays} $R$ if there exists a labeling $e^*$ of $\alpha$ satisfying: 
 \begin{itemize}
  \item $e^*(u,i) \in \{\TS,\TD\}$ if $e(u,i) = \T$; 
  \item $e^*(u,i) = e(u,i)$ if $e(u,i) \neq \T$;
  \item for any distinct $x,y \in \genes$, if $xy \in E$ then $e^*(\lca_G(x,y), \last) \in \{\S,\TS\}$, 
  and otherwise $e^*(\lca_G(x,y), \last) \in \{\D,\TD\}$.
 \end{itemize}
\end{definition}


Note that, \man{if $(G, \alpha)$ and $R$ are known}, there is only one relabeling $e^*$ that \man{ensures that $(G, \alpha)$ displays $R$}.
\mjn{Indeed, if $e(u,i) \neq \T$ then $e^*(u,i) = e(u,i)$ and thus fixed by $(G,\alpha)$; otherwise, $\alpha_i(u)$ is the last element of $\alpha(u)$ and $\alpha_i(u) \notin L(N)$, and thus the value of $e^*(u,i)$ (either $\TS$ or $\TD$) depends on whether $xy \in E$, for any $x,y \in \Gamma$ such that $\alpha_i(u) = \lca_G(x,y)$.}
The question of interest in this paper is, \man{if only $R$ is known}, whether there exists a gene tree that 
displays $R$ and that can be reconciled with \cs{a given network} $N$.


\begin{definition}
Let $N$ be a species network and  $R = (\genes, E)$ a relation graph. We say that $R$ is \emph{consistent} with $N$ (or $N$-consistent)  
if there exists a reconciled gene tree $(G, \alpha)$ with respect to $N$ that displays $R$.
In addition we say that $R$ is \emph{$N$-consistent using $k$ transfers} if $(G, \alpha)$ contains at most $k$ transfers, that is, $e(u, i) = \T$ or $\TL$ for at most $k$ choices of $(u,i)$.
\end{definition}

For an example, see Figure \ref{fig1}: $R$ is consistent using one transfer with $N$  because $(G, \alpha)$ displays $R$ \mj{(setting $e^*(g_5, \last) = \TS$)} and can be reconciled using one transfer (see the reconciliation given above). It is straightforward to see that  $R$ is not consistent using no transfers, thus $R$ is not consistent according to the definition of consistency without xenology  \citep{hernandez2012event,hellmuth2013orthology,lafond2014orthology,hellmuth2015phylogenomics,jones2016}.
\man{It is worth mentioning the question studied in~\cite{hellmuth2019reconciling} can be interpreted as asking whether $R$ is consistent with \emph{some} network $N$.  It turns out that the answer is always yes, albeit a slightly different model is used.}

The main question of interest is to decide whether a set of orthology/paralogy relations can be explained by a gene tree that be reconciled with a given species network.

\medskip
\noindent \textbf{\textsc{Network Consistency (NC):}}\\
\noindent {\bf Input}: A relation graph $R$ and a time-consistent species network $N$.\\
\noindent {\bf Question}: Is $R$ $N$-consistent? \\

We can also consider the minimization version.
It is the same as \textsc{NC}, but we are also given a parameter $k$ and ask 
whether $R$ is $N$-consistent using $k$ transfers.

\medskip
\noindent \textbf{\textsc{Transfer Minimization Network Consistency  (TMNC):}}\\
\noindent {\bf Input}: A relation graph $R$, a time-consistent species network $N$, and an integer $k$.\\
\noindent {\bf Question}: Is $R$ $N$-consistent using at most $k$ transfers? \\

\subsection{Relation graphs and least-resolved DS-trees}

It will be useful to view the problem in terms of a gene tree instead of dealing with relations directly.
Before proceeding with our algorithmic results, we establish the equivalence between relation graphs and \emph{least-resolved DS trees}.  \ml{This relationship was initially established in~\cite{bocker1998recovering}.}  In essence, a DS-tree is simply a gene tree $D$ in which each internal node is labeled $\S$ or $\D$.  This labeling does not have to be valid with respect to any species tree or network.

More formally, a \emph{DS-tree for $\genes$} is a pair $(D,l)$, where $D$ is a 
 rooted tree with $L(D) = \genes$, and
  $l:I(D) \rightarrow  \{\D, \S\}$ 
  is a function labeling each internal node of $G$ as a \emph{duplication} or \emph{speciation}. 
Note that $D$ is not necessarily binary.
The graph $R(D, l) = (\genes, E)$ 
is the relation graph such that for any pair  \man{$\{x,y\}$}
of genes in $\genes$,
if $l(\lca_D(x,y)) = \S$ then $xy \in E$,
and if  $l(\lca_D(x,y)) = \D$ then $xy \notin E$.
We say that $(D, l)$ \emph{displays} a relation graph $R$ if $R(D, l) = R$.

An $l$-contraction in a $DS$-tree $(D, l)$ consists of contracting an arc $(u,v)$ of $D$ with $u ,v \in I(D)$ and  
$l(u) = l(v)$, and assigning the same label to the node resulting from the contraction.
We say that $(D, l)$ is \emph{least-resolved} if no $l$-contraction is possible.
Note that if $(D, l)$ is least-resolved, then it has 
\emph{alternating} duplication and speciation nodes.  That is, 
each child of a speciation node is a duplication or a leaf, and each child of a duplication node is a speciation or a leaf.

A $DS$-tree $(D, l)$ is a \emph{refinement} of another $DS$-tree $(D', l')$ 
if $(D', l')$ can be obtained from $(D, l)$ by a sequence of $l$-contractions.
If $D$ is binary, then  $(D, l)$ is a \emph{binary refinement} of $(D', l')$.
Observe that $l$-contractions do not change $l(\lca_{D'}(x,y))$ for any pair of genes $(x,y)$.
Thus if $(D, l)$ is a refinement of $(D', l')$, then $R(D, l) = R(D', l')$.

It is known that all $DS$-trees that display $R$, if any \mj{exist}, are refinements of the same least-resolved $DS$-tree.

\begin{lemma}[\cite{hellmuth2013orthology,lafond2014orthology}]\label{lem:unique-ds-tree}
Assume that some $DS$-tree displays a relation graph $R$.
Then the least-resolved $DS$-tree $(D, l)$ that displays $R$ is unique.  Moreover, $(D, l)$ can be found in linear time.
\end{lemma}

We now want to relate $DS$-trees with \DTL reconciliations by reinterpreting some internal nodes as transfers.

\begin{definition}\label{def:nreconciable}
Let $N$ be an LGT network and $(D, l)$ a $DS$-tree with $D$ binary.
We say $(D, l)$ is \emph{$N$-reconcilable} if there exists a 
\mj{\DTL reconciliation $\alpha$ between $D$ and $N$}
such that 
for every internal node $u \in I(D)$, the following holds:
\begin{itemize}
    \item
    if $l(u) = \S$, then $e(\alpha_{\last}(u)) \in \{\S, \T\}$;
    
    \item
    if $l(u) = \D$, then $e(\alpha_{\last}(u)) \in \{\D, \T\}$.
\end{itemize}

Moreover, $(D, l)$ is $N$-reconcilable using $k$ transfers if $\alpha$ uses $k$ transfers.

If $D$ is non-binary, we say that $(D, l)$ is $N$-reconcilable \mj{(using $k$ transfers)} if there exists a binary refinement 
$(D', l')$ of $(D, l)$ such that $D'$ is $N$-reconcilable \mj{(using $k$ transfers).}
\end{definition}

Since relation graphs correspond to a unique \man{least-}resolved $DS$-tree, asking about the consistency of a relation graph $R$ is 
equivalent to asking a similar question about a least-resolved DS-tree $(D, l)$ that displays $R$, if it exists (see Appendix for a proof).

\begin{lemma}\label{lem:equiv-ds-tree}
Let $N$ be an LGT network and $R = (\genes, E)$ a relation graph. Then $R$ is $N$-consistent  
(using $k$ transfers)
if and only if 
there exists a DS-tree $(D, l)$ for $\genes$ such that $R(D, l) = R$ and such that $(D, l)$ is $N$-reconcilable (using $k$ transfers).
\end{lemma}

Note that in particular, Lemma ~\ref{lem:equiv-ds-tree}  implies that for $R$ to be $N$-consistent
\mj{for} an LGT network $N$, there must exist a $DS$-tree $(D', l')$ such that $R(D', l') = R$.  Moreover, we may assume that $(D', l')$ is a binary refinement of the unique least-resolved $DS$-tree $(D, l)$ that displays $R$.
By Lemma~\ref{lem:unique-ds-tree}, we can check in linear time whether $(D, l)$ exists, and if so construct it.  Therefore,  
we will often describe an instance of our problem by giving the least-resolved DS-tree $(D, l)$ satisfying $R(D, l) = R$.

\ml{We close this subsection by mentioning that the notion of consistency of a gene tree (or DS-tree) has been studied the other way around.  That is, in~\cite{markin2018solving,gorecki2019feasibility}, we are instead given a species tree and a gene family, and must find a feasible gene scenario under certain constraints.}

\subsection{Basics of parameterized complexity}

We finish this section with some basics of parameterized complexity.
A \emph{parameterized problem} is a language $L \subseteq \Sigma^* \times \mathbb{N}$, where $\Sigma$ is a fixed alphabet and $\Sigma^*$ are the strings over this alphabet. 
A pair $(x,k) \in  \Sigma^* \times \mathbb{N}$ is a {\sc Yes}-instance of a parameterized problem $L$ if $(x,k) \in L$. 
We call the second element $k$ the \emph{parameter} of the instance.
A parameterized problem is \emph{fixed-parameter tractable (FPT)} if there exists an algorithm that decides whether a given instance $(x,k)$ is a {\sc Yes}-instance in time $f(k)\cdot|x|^{O(1)}$, where $f$ is a computable function depending only on $k$; such an algorithm is called an \emph{FPT algorithm}.
The class $W[1]$ is a class of parameterized problems which are strongly believed to not be FPT.
A parameterized problem $L$ is \emph{$W[1]$-hard} if there exists $L' \in W[1]$ such that an FPT algorithm for $L$ would imply an FPT algorithm for $L'$.
For more information we refer the reader to~\cite{DowneyFellows2013}. 



%

\vspace{-2mm}

\section{Hardness of minimizing transfers on LGT networks}\label{sec:w1hardness}

In this section, we consider the \textsc{NC} and \textsc{TMNC} problems.
%
%
%
%
%
%
%
We will show that \textsc{NC} is NP-hard.
Moreover, we will show that the minimization version \textsc{TMNC} is not only NP-hard, but also $W[1]$-hard parameterized by $k$, the number of transfers.
We give a reduction from the following problem, which is known to be 
\mj{NP-hard and $W[1]$-hard with respect to $k$~\citep{FELLOWS200953}:}

\medskip

\noindent \textbf{ $k$-Multicolored Clique:}\\
\noindent {\bf Input}: A graph $H = (V,E)$, a partition of $V$ into color classes $V_1, \dots, V_k$. \\
\noindent {\bf Parameter}: $k$.\\
\noindent {\bf Question}: Is there a clique $C$ in $H$ containing exactly one vertex from each color class $V_i$?

\medskip

The full version of the reduction can be found in the Appendix, but we can sketch the essential ideas here.
We describe the NP-hardness proof -- the $W[1]$-hardness is similar but ensures that the reduction 
is parameterized by $k$. 
We first reduce \textbf{$k$-Multicolored Clique} to a novel intermediate problem, \textsc{Antichain on Trees (ACT)}, then reduce \textsc{ACT} to \textsc{NC}.
\textsc{ACT} is formally defined below, but the intuition is as follows: 
we are given a tree $T$, a set $X$ of elements to place on the nodes of $T$, 
and a weight function $w:X \times V(T) \rightarrow \mathbb{N}_0 \cup \{\infty\}$ indicating the cost of 
placing $x \in X$ on $v \in V(T)$.  We interpret $w(x, v) < \infty$ as ``$x$ can go on $v$'' and 
$w(x, v) = \infty$ as ``$x$ cannot go on $v$''.  Our goal is to place each $x \in X$ on an allowable node 
such that the elements of $X$ are pairwise incomparable (i.e. none is an ancestor of the other). 

\medskip

\noindent \textbf{\textsc{Antichain on Trees (ACT):}}\\
\noindent {\bf Input}: An rooted tree $T$, a set $X$, a cost function $w:X \times V(T) \rightarrow \mathbb{N}_0 \cup  \{\infty\}$.\\
\noindent {\bf Question}: Does there exist an assignment $f:X \rightarrow V(T)$ such that $f(x)$ and $f(y)$ are incomparable in $T$ (that is, neither is an ancestor of the other) for each $x \neq y \in X$, and $w(x,f(x)) < \infty$ for each $x \in X$? \\

\vspace{-2mm}

We call an assignment $f$ an \emph{incomparable assignment} if it satisfies the conditions of an \textsc{ACT} instance.
In the minimization version of \textsc{ACT}, which we call 
\textsc{Minimum Weight Antichain on Trees (MWACT)}, 
we are given a parameter $k$ and ask if there is an incomparable assignment of weight at most $k$. 

\medskip

\noindent \textbf{\textsc{Minimum Weight Antichain on Trees (\mjn{MWACT}):}}\\
\noindent {\bf Input}: A rooted tree $T$, a set $X$, a cost function $w:X \times V(T) \rightarrow \mathbb{N}_0 \cup  \{\infty\}$, and an integer $k$.\\
\noindent {\bf Question}: Does there exist an assignment $f:X \rightarrow V(T)$ such that $f(x)$ and $f(y)$ are incomparable in $T$ (that is, neither is an ancestor of the other) for each $x \neq y \in X$, and such that $\sum_{x \in X} w(x,f(x)) \leq k$? \\


To see the relationship between \textsc{ACT} and \textsc{NC}, consider an \textsc{ACT} instance
$(T, X, w)$.  In the \textsc{NC} setting, $N$ is obtained from $T$ after incorporating some 
specific secondary arcs, and the given relations $R$ have, as their unique least-resolved $DS$-tree $(D, l)$,
a speciation root with $|X|$ children, each child being a duplication corresponding to an element of $X$.
Then being able to place $x \in X$ on $v \in V(T)$ represents ``$\alpha_\last(x) = v$ is possible'', 
i.e. the $x$ node of $D$ is \emph{mappable} onto $v$.
That is, the node $v$ has a directed path to every species present at a leaf below $x$, and the weight $w(x, v)$ is the number 
of transfers required to do so.
To enforce the $\alpha_\last(x)$ to be pairwise incomparable,
we ensure that transfers can only be undertaken by descendants of the $X$ nodes of $D$.
Thus the speciation root of $D$ cannot be explained by any transfer whatsoever, ensuring that its children must be 
incomparable.  We now proceed with the formalization of these ideas, and direct the reader to the Appendix for 
the details of the constructions.

We first show that \textsc{ACT} is NP-hard and \textsc{MWACT} is $W[1]$-hard even under certain restrictions; these will allow us to reduce \textsc{ACT}  to \textsc{NC} and \textsc{MWACT} to \textsc{TMNC}.
The main idea is that the incomparability requirement can be used to create gadgets as subtrees of an \textsc{ACT} or \textsc{MWACT} instance -- if some parent node is assigned to a variable in $X$, then none of its children can be assigned to any variable in $X$. In addition, the  weight function allows to limit the number of places that can be assigned to a given variable. Using these ideas, we can create an instance of \textsc{ACT}, such that an incomparable assignment of finite weight exists if and only if a given instance of {\sc $k$-Multicolored Clique} is a {\sc Yes}-instance.

\begin{lemma}\label{lem:MWACTconstruction}
Let $H = (V = V_1 \cup V_2 \cup \dots \cup V_k, E)$ be an instance of {\sc $k$-Multicolored Clique}.
Then in polynomial time, we can construct an instance $(T,X,w)$ of \textsc{ACT} such that
$(T,X,w)$ has an incomparable assignment of weight $< \infty$ if and only if 
$H$ has a $k$-multicolored clique.
Furthermore, if an incomparable assignment of weight $w< \infty$ exists,
then there exists an incomparable assignment with weight $\le k' = k^2+2k$,
and $(T,X,w)$ satisfies the following properties:

\begin{itemize}
 \item $w(x,v) \in \{0,1,\infty\}$ for all $x \in X, v \in V(T)$;
 \item $w(x,v)= 0$ for exactly one $v$ for each $x \in X$;
 \item if $w(x,v) = 0$ then $w(y,v) = \infty$ for all $y \neq x$;
 \item for any $x \in X$, $u,v \in V(T)$ such that $w(x,u),w(x,v) < \infty$, $u$ and $v$ are incomparable.
\end{itemize}
\end{lemma}

As $(T,X,w)$ is a {\sc Yes}-instance of \textsc{ACT} if and only if the corresponding instance of {\sc $k$-multicolored clique} is a {\sc Yes}-instance, we have that \textsc{ACT} is NP-hard.
Moreover, let \linebreak $(T,X,w,k')$ be the instance of  \textsc{MWACT} with $k' = k^2 + k$ and $T,X,w$ as in Lemma~\ref{lem:MWACTconstruction}. Then Lemma~\ref{lem:MWACTconstruction}  also implies that $(T,X,w,k')$ is a {\sc Yes}-instance of \textsc{MWACT} if and only if the corresponding instance of {\sc $k$-multicolored clique} is a {\sc Yes}-instance. As $k'$ is expressible as a function of $k$, any FPT algorithm for Lemma~\ref{lem:MWACTconstruction} implies a FPT algorithm for {\sc $k$-multicolored clique}. Therefore, as {\sc $k$-multicolored clique} is $W[1]$-hard, so is \textsc{MWACT}. Moreover as $(T,X,w)$ satisfies the properties of Lemma~\ref{lem:MWACTconstruction}, we have the following:

\begin{lemma}
 \textsc{ACT} is NP-hard and \textsc{MWACT} is $W[1]$-hard, even under the following conditions:
 
 \begin{itemize}
 \item $w(x,v) \in \{0,1,\infty\}$ for all $x \in X, v \in V(T)$;
 \item $w(x,v)= 0$ for exactly one $v$ for each $x \in X$;
 \item if $w(x,v) = 0$ then $w(y,v) = \infty$ for all $y \neq x$;
 \item for any $x \in X$, $u,v \in V(T)$ such that $w(x,u),w(x,v) < \infty$, $u$ and $v$ are incomparable.
\end{itemize}
\end{lemma}

%


We next reduce \textsc{ACT} to \textsc{NC}.
The main idea behind this reduction is that every element of $X$ can be represented by a child of the same speciation node in a least-resolved DS-tree.
The tree $T$ can be represented by the distinguished base tree in the species network, and secondary arcs can be added in such a way that, for any \DTL reconciliation, the node corresponding to $x \in X$ can only be mapped to nodes $v$ for which $w(x,v) < \infty$.

\begin{lemma}\label{lemma:NCreduction}
Let $(T,X,w)$ be an instance of \textsc{ACT}, such that $w(x,v) \in \{0,1,\infty\}$ for all $x \in X, v \in V(T)$, $w(x,v)= 0$ for exactly one $v$ for each $x \in X$,
\mj{if $w(x,v) = 0$ then $w(y,v) = \infty$ for all $y \neq x$,}
 and for any $x \in X$, $u,v \in V(T)$ such that $w(x,u),w(x,v) < \infty$, $u$ and $v$ are incomparable.
 
 Then in polynomial time, we can construct both a least-resolved DS-tree $(D, l)$ and a \mjn{time-consistent} LGT network $N$ such that for any integer $k$,
 $(T,X,w)$ has an incomparable assignment of cost at most $k$ if and only if there exists a binary refinement $(D', l')$ of $(D, l)$ such that $(D', l')$ is $N$-reconcilable using at most $2k$ transfers.
\end{lemma}

\medskip

By setting $R = R(D, l)$, Lemma~\ref{lemma:NCreduction} implies that $R$ is $N$-consistent if and only if $(T,W,x)$ has an incomparable assignment of cost $< \infty$, i.e. $(T,W,x)$ is a {\sc Yes}-instance of  \textsc{ACT}. As  \textsc{ACT} is NP-hard (under the restrictions in  Lemma~\ref{lemma:NCreduction}), so is  \textsc{NC}.
Moreover, for any integer $k$, Lemma~\ref{lemma:NCreduction} implies that $R$ is $N$-consistent using at most $k'=2k$ transfers if and only if  $(T,W,x,k)$ is a {\sc Yes}-instance of \textsc{MWACT}. As \textsc{MWACT} is $W[1]$-hard (under the restrictions in  Lemma~\ref{lemma:NCreduction}), so is \textsc{TMNC}.


\begin{theorem}
 \textsc{NC} is NP-hard and \textsc{TMNC} is $W[1]$-hard.
\end{theorem}



\vspace{-2mm}

\section{Dynamic programming for bounded degree DS-trees}\label{sec:dpalgo}
\setlength{\textfloatsep}{2pt}
\begin{algorithm}[!b]

\KwData{A DS-tree $D$, an LGT network $N$}
\KwResult{$\infty$ if $D$ is not $N$-reconcilable, or otherwise the minimum number of transfers}
\caption{minTransferCost($D, N$)}

Initialize $f(g, s) = \infty$ for all $g \in V(D), s \in V(S)$

\For{$g \in V(D)$ in post-order traversal}{
	\For{$s \in V(N)$ in post-order traversal}{
		\uIf{$g$ is a leaf}{
			$f(g, s) = 0$ if $\sigma(g) = s$, otherwise $f(g, s) = \infty$
		}
		\Else{
			$best = \infty$

			\For{$(D', l') \in \mathcal{B}(g)$}{
				$b = reconcileLBR((D', l'), N, s, f)$
				
				\lIf {$b < best$}{$best = b$}
			}
			$f(g, s) = best$
		}
	}
}
return $\min_{s \in V(N)} f(r(D), s)$

\end{algorithm}

In this section, we show that given a \ml{relation graph} $R$ and its least-resolved DS-tree $(D, l)$, if every node of $D$ has degree at most $k$, then 
one can decide if $(D, l)$ is $N$-reconcilable in time 
\mj{$O(2^{k}k!k|V(D)||V(N)|^4)$.}
Moreover, if $(D, l)$ is $N$-reconcilable, our algorithm finds the minimum number of transfers
required by any possible reconciliation.
In particular, if $D$ is binary, then TMNC can be solved in polynomial time. 
\ml{Note that in~\cite{hellmuth2019reconciling}, it is shown that a DS-tree can always be reconciled with some network in a similar reconciliation model, and the authors characterized precisely when a DS-tree  can be reconciled with a given network (although transfers are not studied and, hence, not minimized as we do here).  Let us also mention that in a series of papers~\citep{hellmuth2017biologically,nojgaard2018time,hellmuth2019reconciling}, it is shown how, given a DS-tree with known transfer events but no species phylogeny, one can find a species tree/network that it can be reconciled with.}

The idea of the algorithm is similar to those  of~\cite{Scornavacca2016,kordi2017exact,mykowiecka2017inferring}.
We use dynamic programming over $V(D)$, from the leaves to the root, and when we encounter a non-binary node, we try every way of refining it.   This is a relatively standard procedure, although ensuring a valid reconciliation while minimizing transfers requires care.

For each $g \in V(D)$ and each $s \in V(N)$, we denote by $f(g, s)$ the minimum number of transfers needed 
by a reconciliation
$(D_g, \alpha)$ with respect to $N$ if we require $\alpha_\last(g) = s$ (recall that $D_g$ is the subtree of $D$ rooted at $g$).
If $g$ is a binary node, we try mapping $g_l$ and $g_r$ to every pair of species $s_1$ and $s_2$
that allow $e(g, \last) \in \{l(g), \T \}$, and $f(g, s)$ is the minimum over all possibilities.  For fixed $s, s_1$ (resp. $s_2$), the number of transfers required 
on the branch $(g, g_l)$ (resp. $(g, g_r)$) is the minimum number of secondary arcs on a path from $s$ to $s_1$ (resp. $s_2$).
This path would constitute the sequence $\alpha(g_l)$ (resp. $\alpha(g_r)$).  
  Then $f(g, s)$ can be computed from these values, 
plus those of $f(g_l, s_1)$ and $f(g_r, s_2)$.  
If $g$ is a non-binary node with children $g_1, \ldots, g_k$, we simply try to refine $g$ in every possible way, then do as in the binary case.  In such a binary refinement $B$ of $g$, we may treat the $g_1, \ldots, g_k$ nodes of $B$ as leaves 
and use the previously computed $f(g_i, s')$ values for each $(g_i, s')$ pair.  
Let us turn to the algorithmic details.


\setlength{\textfloatsep}{2pt}
\begin{algorithm}[!t]

\KwData{A binary DS-tree $(D', l')$ which is an LBR of some subtree of $D$, an LGT network $N$, the desired species $s$ for $r(D')$, a cost function $f$ on the leaves of $D'$}
\KwResult{The minimum cost to reconcile $D'$ with $N$ such that $\alpha_\last(r(D')) = s$}
\caption{reconcileLBR($D', N, s, f)$}

Set $f' = f$  (we maintain temporary costs $f'$ for $D'$)

\For{$g \in I(D')$ in post-order traversal}{
   \For{$s' \in V(N)$ in post-order traversal}{
      \uIf{$l'(g) = \S$}{
      	 \uIf{\man{$s'$ has two children and } $(s', s'_l), (s', s'_r) \in E_p$}{

        $cost12 = \min_{(s_1, s_2) \in P(s'_l) \times P(s'_r)} (f'(g_l, s_1) + t(s'_l, s_1) + f'(g_r, s_2) + t(s'_r, s_2))$
        
        $cost21 = \min_{(s_1, s_2) \in P(s'_l) \times P(s'_r)} (f'(g_r, s_1) + t(s'_l, s_1) + f'(g_l, s_2) + t(s'_r, s_2))$

		$f'(g, s') = \min (cost12, cost21 ) $		 \label{line:the-s-case}
		 }
		\ElseIf{$s'$ is the tail of a secondary arc $(s', s'')$ ($s'' \in \{s'_l, s'_r\}$)}{

            $cost12 = 1 + 
            \min_{(s_1, s_2) \in P(s') \times P(s'')} (f'(g_l, s_1) + t(s', s_1) + f'(g_r, s_2) + t(s'', s_2))$

            $cost21 = 1 + 
            \min_{(s_1, s_2) \in P(s') \times P(s'')} (f'(g_r, s_1) + t(s', s_1) + f'(g_l, s_2) + t(s'', s_2))$

            $f'(g, s') = min(cost12, cost21)$ \label{line:the-t-case}
            

		}
      }
      \ElseIf{$l'(g) = \D$}{
            $f'(g, s') =  \min_{(s_1, s_2) \in P(s') \times P(s')} (f'(g_l, s_1) + t(s', s_1) + f'(g_r, s_2) + t(s', s_2))$ \label{line:the-d-case}
     }

   }
}

return $f'( r(D'), s)$

\end{algorithm}

Let $g \in I(D)$ with children $g_1, \ldots, g_k$.  
A binary DS-tree $(D', l')$ with root $g$ and leafset $g_1, \ldots, g_k$ such that $l'(g') = l(g)$ for every $g' \in I(D')$ will be called a \emph{local binary refinement}
of $g$ (we write LBR for short).  We denote by $\mathcal{B}(g)$ the set of possible LBRs of $g$.
For $s \in V(N)$, denote by $P(s)$ the set of vertices of $N$ that can be reached by some directed path starting from $s$, 
and let $t(s, s')$ denote the minimum number of secondary arcs necessary to go from $s$ to $s'$
(note that $t(s, s')$ is easy to compute using weighted shortest path algorithms).
We let $t(s, s') = \infty$ if there is no path from $s$ to $s'$.

The algorithm $minTransferCost$ traverses $D$ in a post-order traversal and, 
for each node $g$ and each LBR $D'$ in $\mathcal{B}(g)$, 
calls $reconcileLBR$ to reconcile $D'$.  Note that in the case that $g$ is binary, 
only one LBR is tested, namely the tree with two leaves $g_l$ and $g_r$.


The proof of correctness can be done by induction over the height of $D_g$ and can be found in the Appendix.
For the complexity, we first compute the all-pairs shortest paths 
in $N$ in time $O(|V(N)|^3)$ (this is only done once and will not contribute to the final complexity). 
It is known that the number of binary trees on $k$ leaves 
is 
\mj{$(2k-3)!! = O(2^kk!)$~\citep{felsenstein2004inferring}}
which bounds the size of each set of LBRs.  
The main algorithm computes $\mathcal{B}(g)$ up to $|V(D)||V(N)|$ 
times.  Each member of each $\mathcal{B}(g)$ results in a call to $reconcileLBR$, which is done with a tree $D'$ on 
at most $k$ leaves.  Then in this subroutine for each $(g, s)$ pair with $g \in V(D')$
 and $s \in V(N)$, $O(|V(N)|^2)$ pairs of the form $(s_1, s_2)$ are tested -- 
 this takes time $O(k|V(N)|^3)$.
 The total time is thus 
 $O(2^{k}k!k|V(D)||V(N)|^4)$.
\ml{The space taken by the algorithm is $O(|V(D)||V(N)| + |V(N)|^2)$.  To see this, observe that $O(|V(N)|^2)$ space is needed to store the aforementioned all-pairs shortest path values and $O(|V(D)||V(N)|)$ space is needed for the $f(g, s)$ values.  Each enumerated $(D', l') \in \mathcal{B}(g)$ takes space $O(k) = O(|V(D)|)$, which does not add to the space complexity if only the current such $(D', l')$ is kept in memory at all time. Also, one can check that $reconcileLBR$ can be done without additional space (the $P(s)$ sets can be computed on the fly each time when needed).}

\begin{theorem}\label{thm:dp-algo-is-ok}
Algorithm $minTransferCost$ is correct.  Moreover, it runs in time \linebreak 
$O(2^{k}k!k|V(D)||V(N)|^4)$ \ml{and space $O(|V(N)||V(D)| + |V(N)|^2)$}.
\end{theorem}

\ml{Note that while we focused on minimizing the contribution of the $k$ parameter in the above algorithm, it is plausible that techniques developed for similar dynamic programming algorithms in~\cite{kordi2017exact,mykowiecka2017inferring} could help reduce the $|V(D)||V(N)|^4$ portion of the complexity.  In essence, a factor of $|V(N)|^2$ is saved in~\cite{kordi2017exact,mykowiecka2017inferring} by defining $f(g, s)$ as the best cost of a reconciliation in which $\alpha_{\last}(g)$ is mapped to any node reachable from $s$ (instead of requiring $s$ itself), which avoids having to minimize over all reachable pairs $(s_1, s_2)$ for every node of $D$ as in our algorithm.}

\vspace{-3mm}

\section{With unknown transfer highways}\label{sec:unknownhighways}

\vspace{-2mm}

The set of secondary arcs on a species network cannot always be known with confidence.  In fact, reconciliation is sometimes used to infer such arcs on a given species tree~\citep{THL2011}. 

In this section, we remove the assumption that transfer arcs are known.  We are given a species \emph{tree} $S$ with $|L(S)| > 1$, and 
the secondary arcs $E_s$ are to be determined in a time-consistent manner.  The question is whether,
for a relation graph $R$, there is 
a species network $N$ with base tree $T_0(N) = S$ such that $R$ is $N$-consistent.

\begin{definition}
Let $S$ be a species tree.
We say \man{that a relation graph} $R$ is \ml{\emph{$S$-base-consistent}} (using $k$ transfers) if
there exists a time-consistent \cs{LGT} network $N$ such that 
$T_0(N) = S$ and 
$R$ is  $N$-consistent (using $k$ transfers).
\end{definition}

We will show that a relation graph $R$ is always \ml{\emph{$S$-base-consistent}}, 
provided there is a $DS$-tree $(D, l)$ that displays $R$.  In fact, we prove that \emph{any} binary $DS$-tree
can be made to ``agree'' with any species tree, no matter how inconsistent they appear to be (provided that each $DS$-tree leaf can be mapped to a corresponding species tree leaf).


Beforehand, we can easily establish the equivalence between relation graphs and $DS$-trees as we did for $N$-consistency.
We say that a $DS$-tree $(D, l)$ is \ml{\emph{$S$-base-reconcilable}} (using $k$ transfers) if there exists a 
time-consistent species network $N$ 
such that $T_0(N) = S$ and $(D, l)$ is $N$-reconcilable (using $k$ transfers).

\begin{lemma}\label{lem:equiv-ds-tree-unknown}
Let $R$ be a relation graph and $S$ be a species tree.
Then $R$ is \ml{$S$-base-consistent} (using $k$ transfers) if and only if 
there exists a least-resolved $DS$-tree $(D, l)$ \mj{that displays $R$} and a binary refinement $(D', l')$ of $(D, l)$ such that
$(D', l')$ is \ml{$S$-base-reconcilable} (using $k$ transfers).
\end{lemma}


To show that any $DS$-tree $(D, l)$ is \ml{$S$-base-reconcilable}, we add to $S$ a set of secondary arcs $E_s$ of size 
\mj{$O(h(D)|V(S)|^2)$,} \man{where $h(D)$ is the height of $D$ (see below)}.
\ml{We then obtain a reconciliation $\alpha$ in which $e(\alpha_{\last}(u)) = \mathbb{T}$ for every internal node $u$ of $D$}, which might be necessary in some cases.
For a node $v \in V(D)$, we denote by $d(v)$ the 
\emph{depth} of $v$, which is the number of edges on the path between $v$ and $r(D)$.
The \emph{height} of $D$, denoted $h(D)$, is the maximum depth of a node of $D$.
Let $m = |L(S)|$, and let $(s_1, \ldots, s_m)$ be an arbitrary ordering of $L(S)$.   
Recall that for $i \in [m]$, $s_i$ is a leaf of $S$, and that $p(s_i)s_i$ refers to the edge from the parent of $s_i$ to $s_i$.
We construct the network $N(D)$ from $S$ using the following algorithm:

\setlength{\textfloatsep}{2pt}
\begin{algorithm}[H]
\caption{constructNetwork($D, S$)}
\label{algo:everything-consistent}

	\For{$d = 0$ to $h(D) + 1$}{
		\For{$i = 1$ to $m$}{
			\For{$j = 1$ to $m$, $j \neq i$}{
				Subdivide the arc $p(s_i)s_i$, creating a donor node $\sdon{i}{j}{d}$ \;
				Subdivide the arc $p(s_j)s_j$, creating a receiver node $\srec{j}{i}{d}$ \;
				Add the secondary arc $(\sdon{i}{j}{d}, \srec{j}{i}{d})$ to $E_s$	\;	
			}
		}
	}

\end{algorithm}

Thus we add every transfer from the $s_1$ branch to the $s_i$ branch with $i \neq 1$, 
then every transfer from the $s_2$ branch to the other $s_i$ branches, and so on, and 
repeat this process $h(D) + 2$ times.  Note that $p(s_i)$ changes with each subdivision.
It is not hard to see that $N(D)$ is time-consistent, 
\man{since each time we insert a new arc $(x, y)$, its two endpoints $x$ and $y$ are 
below every other previously inserted node.}

\begin{lemma}\label{lem:all-ds-trees-reconcilable}
Let $(D, l)$ be any binary $DS$-tree and let $N := N(D)$ be the species network obtained from $S$ after
applying Algorithm~\ref{algo:everything-consistent}. 
Then $(D, l)$ is $N$-reconcilable.
\end{lemma}

The detailed proof can be found in the Appendix.  The idea is that each $v \in I(D)$ at depth 
$d(v)$ has the 
secondary edge $(\sdon{i}{j}{d(v)}, \srec{j}{i}{d(v)})$ 
\mj{at its disposal.}
It can be shown that for any $v \in I(D)$ and any distinct $s_i, s_j \in L(N)$, 
$D_v$ can be reconciled with $N$ such that $\alpha(v) = (\sdon{i}{j}{d(v)})$. 
The idea is illustrated in Figure~\ref{fig:dumb-reconciliation}.  The highest node of $D$ is mapped to a highest donor node of $N$, and the descendants transfer back and forth, each time being mapped to a deeper donor node of $N$.

\begin{figure}
    \centering
    \includegraphics[width=.8\textwidth]{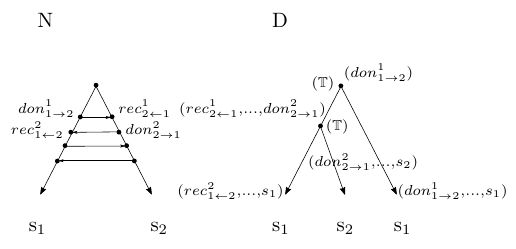}
    \caption{On the left, a species tree with two leaves, with \mj{the horizontal} arcs inserted by the algorithm 
    On the right, a DS-tree in which every internal node is labeled $\S$ initially (not shown) which each become a transfer node which we can then label $\TS$ (the leaves of $D$ depict the species of the gene, we omit giving each gene a name).}
    \label{fig:dumb-reconciliation}
\end{figure}


\begin{theorem}\label{thm:sconsistent-iff-dstree}
A relation graph $R$ is \ml{$S$-base-consistent} if and only if
there exists a $DS$-tree $(D, l)$ that displays $R$.

Therefore, deciding if a relation graph $R$ is \ml{$S$-base-consistent} can be done in polynomial time.
\end{theorem}

Thus, unlike $N$-consistency, 
deciding \ml{$S$-base-consistency} of $R$ can be done quickly by verifying if $R$ admits a $DS$-tree.
However, the explanation of $R$ resulting from the above algorithm 
will produce scenarios with many transfers, all of which are located between a leaf and its parent.  Thus it makes sense to ask
if there is a scenario with at most $k$ transfers.
This problem is closely related to reconciling a gene tree with a species tree while minimizing the number of transfers.  In~\cite{THL2011}, this problem is shown to be NP-hard.

In fact, we present a reduction for minimizing transfers that is very similar in spirit to the one given in~\cite{THL2011}.
There are, however, many differences between their problem and ours that prevent us from using the previous reduction as a black box for our purposes.
First, our definition of reconciliation is different, and in particular, 
in~\cite{THL2011}, transfer-loss events are not allowed.
Also, in the $DS$-tree formulation derived from Lemma~\ref{lem:equiv-ds-tree}, we are given which 
nodes of $D$ must be speciations, and which must be duplications.
Finally, the authors require that the output network contains no directed cycle, 
whereas we require time-consistency, which is more restrictive.
\man{We invite the interested reader to consult the last \cs{section of the Appendix for} details.}




\begin{theorem}\label{thm:hard-unknown-highways}
The problem of deciding if a relation graph $R$ is \ml{$S$-base-consistent} using $k$ transfers 
is NP-hard, even if the least-resolved $DS$-tree $(D, l)$ for $R$ is binary.
\end{theorem}


\vspace{-2mm}

\section{\mj{Discussion}}

\vspace{-2mm}

In this work, we have shown that consistency of relations in the presence of transfers is
computationally hard to deal with, making its application difficult in practice.  One possible avenue would be to attempt to apply our FPT algorithm to real datasets.
A similar algorithm was reported in~\cite{kordi2017exact} to be able to handle nodes with up to
8 children, so a next step would be to check the size of non-binary nodes of DS-trees. 
It would also be interesting to study the problem of error correction of relations
in the presence of transfers - although this is almost certainly NP-hard, 
approximation or FPT algorithms may be applicable.

\subsection*{Acknowledgment}

Version 6 of this preprint has been peer-reviewed and recommended by Peer Community In Mathematical and Computational Biology (\url{https://doi.org/10.24072/pci.mcb.100009}).

\subsection*{Data, script and code availability}

There is no data, script or code associated with the work presented in this paper.

\subsection*{Conflict of interest disclosure}

The authors of this preprint declare that they have no financial conflict of interest with the content of this article.
Celine Scornavacca is one of the PCI Math Comp Biol recommenders.

\vspace{-4mm}

\printbibliography[notcategory=ignore]

\section*{Appendix}
Here we include the details of the proofs that were left out of the main text.

\vspace{4mm}


\noindent
\textbf{Lemma \ref{lem:equiv-ds-tree}.} 
\emph{Let $N$ be an LGT network and $R = (\genes, E)$ a relation graph. Then $R$ is $N$-consistent  
(using $k$ transfers)
if and only if 
there exists a DS-tree $(D, l)$ for $\genes$ such that $R(D, l) = R$ and such that $(D, l)$ is $N$-reconcilable (using $k$ transfers).}

\begin{proof}
($\Rightarrow$) Let $(G, \alpha)$ be a gene tree reconciled with $N$ such that $(G, \alpha)$ displays $R$ using $k$ transfers, and let $e^*$ be a labeling such that $e^*(u,i) \in \{\TS,\TD\}$ if $e(u,i) = \T$, $e^*(u,i) = e(u,i)$ if $e(u,i) \neq \T$, and if $xy \in E$ then $e^*(\lca_G(x,y), \last) \in \{\S,\TS\}$, 
  and otherwise $e^*(\lca_G(x,y), \last) \in \{\D,\TD\}$.
  
  Now define a binary DS-tree $(D,l)$ as follows. 
  Let $D = G$, and let $l(u) = \S$ if \linebreak $e^*(\alpha_\last(u)) \in \{\S,\TS\}$, and $l(u) = \D$ otherwise (in which case $e^*(\alpha_\last(u)) \in \{\D,\TD\}$).
  Observe that by definition of $e^*$, 
  if $l(\lca_{D}(x,y)) = \S$ then $xy \in E$,
  and if  $l(\lca_{D}(x,y)) = \D$ then $xy \notin E$.
  Thus we have that $R = R(D, l)$.
  Also, note that $(D, l)$ is $N$-reconcilable using $k$ transfers, 
  since $\alpha$ satisfies the conditions of 
  Definition~\ref{def:nreconciable}.

  ($\Leftarrow$): let $(D, l)$ be a $DS$-tree such that $R(D, l) = R$.  
  Note that $D$ is not necessarily binary.
  Let $(D', l')$ be a binary refinement 
  of $(D, l)$ such that $(D', l')$ is $N$-reconcilable (such a refinement is assumed to exist by the lemma statement and by the definition of $N$-reconcilable for non-binary gene trees).
  Since $(D', l')$ is $N$-reconcilable, there exists $\alpha'$ such that $(D', \alpha')$ is a reconciled gene tree with respect to $N$ such that 
  for every $u \in I(D')$, $l'(u) = \S$ implies $e(\alpha_{\last}(u)) \in \{\S, \T\}$ 
  and $l'(u) = \D$ implies $e(\alpha_{\last}(u)) \in \{\D, \T\}$.
  Define $e^*$ as follows:  
  if $e(u, i) \neq \T$, then $e^*(u, i) = e(u, i)$;
  otherwise if $e(u, i) = \T$, if $l'(u) = \S$ then $e^*(u, i) = \TS$ and if $l'(u) = \D$ then $e^*(u, i) = \TD$.
  Note that no additional transfer is created in this manner, and hence $e^*$ still uses $k$ transfers.
  Also, for any pair of distinct genes $x, y \in \genes$ with $u = \lca_{D'}(x, y)$,
   $l'(u) = \S$ implies $e^*(\alpha_{\last}(u)) \in \{\S, \TS\}$
   and $l'(u) = \D$ implies $e^*(\alpha_{\last}(u)) \in \{\D, \TD\}$.
  It follows that $(D', \alpha')$ display $R$.
\end{proof}


\noindent
\textbf{Lemma \ref{lem:MWACTconstruction}.} 
\emph{Let $H = (V = V_1 \cup V_2 \cup \dots \cup V_k, E)$ be an instance of {\sc $k$-Multicolored Clique}.
Then in polynomial time, we can construct an instance $(T,X,w)$ of \textsc{ACT} such that
$(T,X,w)$ has an incomparable assignment of weight $< \infty$ if and only if 
$H$ has a $k$-multicolored clique.
Furthermore, if an incomparable assignment of weight $w< \infty$ exists,
then there exists an incomparable assignment with weight $\le k' = k^2+2k$,
and $(T,X,w)$ satisfies the following properties:
\begin{itemize}
 \item $w(x,v) \in \{0,1,\infty\}$ for all $x \in X, v \in V(T)$;
 \item $w(x,v)= 0$ for exactly one $v$ for each $x \in X$;
 \item If $w(x,v) = 0$ then $w(y,v) = \infty$ for all $y \neq x$;
 \item for any $x \in X$, $u,v \in V(T)$ such that $w(x,u),w(x,v) < \infty$, $u$ and $v$ are incomparable.
\end{itemize}}
\begin{proof}
 
{\bf Construction of \textsc{ACT} instance:}

Let $H = V = (V_1 \cup V_2 \cup \dots \cup V_k, E)$ be an instance of {\sc $k$-Multicolored Clique}.
We now construct a tree $T$ together with a set $X$ and cost function $w:X \times V(T) \rightarrow \mathbb{N}_0 \cup  \{\infty\}$.
For each element $x \in X$, there will be a single ``in''-element $x\_in$ of $V(T)$, for which $w(x,x\_in) = 0$.
There will also be some number of ``out''-elements $v$ for which $w(x,v) = 1$.

We begin by describing $T$.
$T$ is made up of a series of subtrees, each of which will act as a gadget in our reduction from {\sc $k$-Multicolored Clique}.
Every subtree consists of a root with several leaves as children.

The subtrees of $T$ are as follows:

\begin{itemize}
 \item A tree {\bf Start}, with root $s\_in$ and children $class\_i\_in$ for each $i \in [k]$;
 \item For each $i \in [k]$, $v \in V_i$, a tree {\bf Choose\_$v$}, with root $v\_in$, and children $class\_i\_out\_v$, together with $u\_to\_i\_out\_v$ for each $u \in V \setminus V_i$ such that $uv \in E$;
  \item For each $i \in [k]$, $v \in V_i$, a tree {\bf Cover\_$v$}, with root $v\_out$, and children  $count\_v\_in$ , together with $v\_to\_j\_in$ for each $j \neq i \in [k]$.
  \item For each $i \in [k]$, 
  a singleton tree consisting of the node $count\_i\_out$.
\end{itemize}

See Figure~\ref{fig:clique_to_mwact}.
Finally we add a root node whose children are the roots of all the subtrees given above.
This concludes our construction of $T$.

\begin{figure*}[!h]
\begin{center}
\includegraphics[width= 1.00 \textwidth]{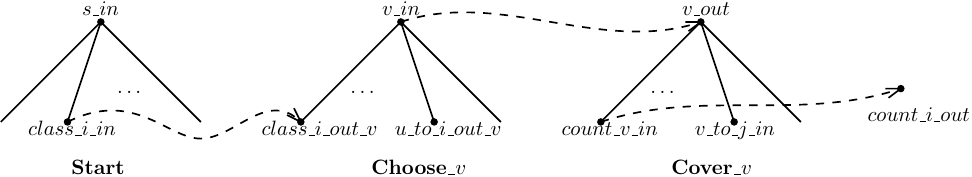}
\caption{Figures used in the reduction from  {\sc $k$-Multicolored Clique} to \textsc{ACT}. Dashed lines represent some of the relations between nodes: If an assignment $f$ does not assign $f(class\_i) = class\_i\_in$, then it must assign $f(class\_i) = class\_i\_out\_v$ (for some $v \in V_i$). Similarly if $f$ does not assign $f(v) = v\_in$, 
then it must assign $f(v) = v\_out$. If $f$ does not assign $f(count\_v) = count\_v\_in$, then it must assign $f(count\_v) = count\_i\_out$.
Note also that if $f$ does not assign $f(v\_to\_j) = v\_to\_j\_in$, then it must assign $f(v\_to\_j) = v\_to\_j\_out\_u$ for some $u \in V_j$ adjacent to $v$, though that relation is not depicted here. 
} 
\label{fig:clique_to_mwact}
\end{center}
\end{figure*}

The set $X$ contains all vertices from $V$. In addition it contains a `start' element $s$, an element $class\_i$ for each $i \in [k]$, an element  $count\_v$ for each $v \in V$, and an element $v\_to\_j$ for each $v \in V_i$ and $j \neq i \in [k]$.

The cost function $w:X \times V(T) \rightarrow \mathbb{N}_0 \cup  \{\infty\}$ is defined as follows:
For each $i \in [k], v\in V_i$ and $j \neq i \in [k]$, set $w(s,s\_in)= w(class\_i, class\_i\_in)= w(v,v\_in)=w(count\_v, count\_v\_in)=w(v\_to\_j, v\_to\_j\_in)=0$.
For each $i \in [k]$ and $v \in V_i$, set $w(class\_i, class\_i\_out\_v) = 1$, set $w(v,v\_out) = 1$, and set $w(count\_v, count\_i\_out) = 1$.
(Note that there are therefore multiple elements $x \in X$ for which  $w(x, count\_i\_out) = 1$.)
Finally, for each $i \in [k]$ and $v \in V_i$, and each edge $uv \in E$ with $u \in V_j, j\neq i\in[k]$, set $w(v\_to\_j, v\_to\_j\_out\_u)=1$.
For all other $x \in X$ and $v \in V(T)$, set $w(x,v) = \infty$.

 This concludes our construction of our \textsc{ACT} instance $(X,T,w)$.
  The construction can be done in polynomial time.
 We observe that by construction, $w(x,v) \in \{0,1,\infty\}$ for all $x \in X, v \in V(T)$, 
  $w(x,v)= 0$ for exactly one $v$ for each $x \in X$,
 and if $w(x,v) = 0$ then $w(y,v) = \infty$ for all $y \neq x$.
  To see that $u$ and $v$ are incomparable for  $x \in X$, $u,v \in V(T)$ such that $w(x,u),w(x,v) < \infty$, observe that each subtree in the construction
  contains at most one node $z$ with $w(x,z) < \infty$ for each $x \in X$.

  It remains to show that $(T,X,w)$ has an incomparable assignment of weight $< \infty$ if and only if 
$H$ has a $k$-multicolored clique and that if an incomparable assignment of weight $w< \infty$ exists,
then there exists an incomparable assignment with weight $\le k'$.
To do this, we will first show that the existence of a $k$-multicolored clique implies the existence of  an incomparable assignment with weight $\le k'$,
and then show that the existence of an incomparable assignment of weight $w< \infty$ implies the existence of a $k$-multicolored clique.

{\bf $k$-multicolored clique implies assignment of weight $\le k'$:}

First suppose that a $k$-multicolored clique $C$ exists, and let $v_i$ denote the single vertex in $C \cap V_i$, for each $i \in [k]$.
Let $f:X \rightarrow V(T)$ be defined as follows:
Set $f(s) = s\_in$.
For each $i \in [k]$, set $f(class\_i) = class\_i\_out\_{v_i}$.
For each $i \in [k]$, set $f(v_i) = {v_i}\_out$, and for all other $v \in V$ set $f(v)=v\_in$.
For each $i \in [k]$, set $f(count\_{v_i}) = count\_i\_out$, and for all other $v \in V$ set $f(v)=count\_v\_in$.
For each $i \in [k]$, $j \neq i \in [k]$, set $f({v_i}\_to\_j) = {v_i}\_to\_j\_out\_{v_j}$
(note that ${v_i}\_to\_j\_out\_{v_j}$ exists because $v_j \in V_j$ and $v_i,v_j$ are adjacent). 
For all other $v \in V_i$, set $f(v\_to\_j) = v\_to\_j\_in$.

Observe that $\sum_{x\in X}w(x,f(x)) = k+ k + k + k(k - 1) = k^2+2k = k'$.
It remains to show that $f(x)$ and $f(y)$ are incomparable for each $x \neq y \in X$.
As each of the subtrees described above are incomparable, it is enough to show that for each subtree, there are no comparable $y,z$ with $y,z$ assigned to different elements of $X$.

In {\bf Start}, the root  $s\_in$ is assigned  but none of the children $class\_i\_in$ are assigned, so we have no comparable assigned nodes.

In {\bf Choose\_$v$}, if $v = v_i$ for some $i \in [k]$, then the root ${v_i}\_in$ is not assigned, and as all other nodes are children of  ${v_i}\_in$, there are no comparable assigned nodes.
For all other $v$ in class $V_i$, the root ${v_i}\_in$ is assigned. However, the child $class\_i\_out\_v$ is not assigned (as $class\_i$ is assigned to $class\_i\_out\_{v_i}$), and  the other children $u\_to\_i\_out\_v$ are not assigned ($u\_to\_i\_out\_v$ is only assigned if $v = v_i, u = v_j$ for some $i \neq j \in [k]$).

In {\bf Cover\_$v$}, if $v = v_i$ for some $i \in [k]$, then the root ${v_i}\_out$ is assigned, but none of its children ${v_i}\_to\_j\_in$ or $count\_{v_i}\_in$ are assigned, as ${v_i}\_to\_j$ is assigned to ${v_i}\_to\_j\_out\_{v_j}$ and $count\_{v_i}$ is assigned to $count\_i\_out$.
For other $v \in V$, the root $v\_out$ is not assigned, and as all other nodes are children of  $v\_out$, there are no comparable assigned nodes.

The nodes $count\_i\_out$ are the only nodes in $T$ that may be assigned to more than one element of $X$. However, by definition of $f$ we have that for each $i \in [k]$, $count\_{v_i}$ is the only element assigned to $count\_i\_out$. 

As $\sum_{x\in X}w(x,f(x)) \le k'$ and $f(x),f(y)$ are incomparable for all $x\neq y \in X$, we have that  $(X,T,w,k')$ is a {\sc Yes}-instance, as required.

{\bf Assignment of finite weight implies $k$-multicolored clique:}

Suppose $f:X \rightarrow V(T)$ is an incomparable assignment with $\sum_{x \in X}w(x,f(x)) < \infty$.

Note that $f(s) = s\_in$, as there is no other node $z$ for which $w(s,z)< \infty$.
It follows that $f(class\_i) \neq class\_i\_in$ for each $i \in [k]$.
Therefore $f(class\_i) = class\_i\_out\_v$ for some $v \in V_i$.
Denote this $v$ by $v_i$.
As $class\_i\_out\_{v_i}$ is a child of ${v_i}\_in$ in  {\bf Choose\_$v_i$}, we must have that $f(v_i) \neq {v_i}\_in$, and so instead  $f(v_i) = {v_i}\_out$.
As ${v_i}\_out$ is the root of  {\bf Cover\_$v_i$}, it follows that for each $j \neq i \in [k]$, we cannot have $f({v_i}\_to\_j) = {v_i}\_to\_j\_in$.
Therefore $f({v_i}\_to\_j) = {v_i}\_to\_j\_out\_u$ for some $u \in V_j$ adjacent to $v_i$.
Denote this $u$ by $u_{ij}$.

It remains to show that $u_{ij} = v_j$ for each $i\neq j \in [k]$, as this implies that  $v_1, \dots, v_k$ form a clique.
As $f({v_i}\_to\_j) = {v_i}\_to\_j\_out\_{u_{ij}}$ is a child  of ${u_{ij}}\_in$ in {\bf Choose\_$u_{ij}$}, we must have that $f(u_{ij}) \neq {u_{ij}}\_in$, and so instead  $f(u_{ij}) = {u_{ij}}\_out$.
As $count\_u_{ij}\_in$ is a child of ${u_{ij}}\_out$ in {\bf Cover\_$u_{ij}$}, we must have that $f(count\_u_{ij}) \neq  count\_u_{ij}\_in$ and so instead $f(count\_u_{ij}) =  count\_j\_out$ (recall that $u_{ij} \in V_j$).
By a similar argument, since $f(v_j) = {v_j}\_out$ we also have $f(count\_v_j) =  count\_j\_out$.
But then $f$ is not an incomparable assignment unless $u_{ij} = v_j$ (since $f(count\_u_{ij})$ and $f(count\_v_j)$ are the same node, and therefore comparable).
Therefore we must have that $u_{ij} = v_j$ for all $i\neq j \in [k]$, as required.
\end{proof}

\noindent
\textbf{Lemma \ref{lemma:NCreduction}.} 
\emph{ Let $(T,X,w)$ be an instance of \textsc{ACT}, such that $w(x,v) \in \{0,1,\infty\}$ for all $x \in X, v \in V(T)$, $w(x,v)= 0$ for exactly one $v$ for each $x \in X$,
\mj{if $w(x,v) = 0$ then $w(y,v) = \infty$ for all $y \neq x$,}
 and for any $x \in X$, $u,v \in V(T)$ such that $w(x,u),w(x,v) < \infty$, $u$ and $v$ are incomparable.
 \\
 Then in polynomial time, we can construct a least-resolved DS-tree $(D, l)$ and  \mjn{time-consistent} LGT network $N$ such that for any integer $k$,
 $(T,X,w)$ has an incomparable assignment of cost at most $k$ if and only if there exists a binary refinement $(D', l')$ of $(D, l)$ such that $(D', l')$ is $N$-reconcilable using at most $2k$ transfers.}

\begin{proof}
Let $(T,X,w)$ be an instance of $ACT$ satisfying the specified properties.
We begin by adjusting $T$ to ensure that it is binary.
If an internal node $u$ has a single child, we add an additional child of $u$ as a leaf of the tree.
If $u$ has more than two children, we refine $u$ into a binary tree with the same leaf set (treating $u$ as the root of this binary tree).
For any new node $v$ introduced in this way, we set $w(x,v)=\infty$ for all $x \in X$.
Observe that for the resulting tree $T'$, two nodes $u,v \in V(T)$ are incomparable in $T$ if and only if they are incomparable in $T'$.
Thus, changing $T$ in this way gives us an equivalent instance.

So we may now assume that $T$ is binary.
We next describe how to construct a least-resolved DS-tree $(D, l)$ .

Let $\Gamma$ be a set of genes as follows.
For each $x\in X$, $\Gamma$ contains two new genes $x\_left$ and $x\_right$.
Let $\Sigma$ contain species $spec\_x\_left$ and $spec\_x\_right$ for each $x\in X$, with $\sigma(x\_left) = spec\_x\_left$, $\sigma(x\_right)=spec\_x\_right$.

Let the DS-tree $(D, l)$ contain a speciation node $r$ as the root, and let $\{gene\_x: X\}$ be the set of children of $r$.
For each $x \in X$, let $gene\_x$ be a duplication node with children $x\_left$ and $x\_right$.
Note that $(D, l)$ is a least-resolved DS-tree.

We next describe how to construct the LGT network $N$, beginning with the distinguished base tree $T_0(N)$.
Initially, let $T_0(N)=T$, the input tree of our $ACT$ instance (in its binary version).
To avoid confusion with the MWACT instance later, we rename each node $v \in V(T)$  to $spec\_v$.
In addition, for each $x \in X$ 
let $u_x$ be the unique node in $T$ for which $w(x,u_x)=0$, with $spec\_u_x$ the corresponding node in $N$.

Now for each $v \in V(T)$,
we will add $spec\_v\_left$ and $spec\_v\_right$ as descendants \mj{(not necessarily children)} of $spec\_v$, as follows.
If $spec\_v$ is a leaf in $T_0(N)$, then add $spec\_v\_left$ and $spec\_v\_right$  as children of $spec\_v$.
Otherwise, add $spec\_v\_left$ and $spec\_v\_right$ as descendants of different children of $spec\_v$.
(This can be be done by subdividing any arc incident to leaf descended from a given child of $spec\_v$, and adding $spec\_v\_left$ or $spec\_v\_right$ as a child of the newly added node).
Observe that after $spec\_v\_left$ and $spec\_v\_right$ have been added, $spec\_v$ is the least common ancestor of $spec\_v\_left$ and $spec\_v\_right$. Furthermore this process does not change the least common ancestor of any pair of leaves.
Therefore, after doing this process for each $v \in V(T)$, we will have that for every $v \in V(T)$, $spec\_v$  is the least common ancestor of $spec\_v\_left$ and $spec\_v\_right$.
When $v = u_x$ for some $x \in X$, we also denote $spec\_v\_left$ and $spec\_v\_right$ by $spec\_x\_left$ and $spec\_x\_right$ respectively.

This completes the construction of the distinguished base tree; now we describe how to add secondary arcs.
For each $x \in X$ and each $v \in V(T)$ with $w(x,v)=1$, we do the following.
Add a new tail node between $spec\_v\_left$ and its parent, add a new head node between $spec\_x\_left$ and its parent, and add an arc from the tail to the head as a secondary arc.
Similarly, add  a new tail node between $spec\_v\_right$ and its parent, and add a new head node between $spec\_x\_right$ and its parent, and add an arc from the tail to the head as a secondary arc.
Observe that after this, $spec\_v$ has paths to $spec\_x\_left$ and $spec\_x\_right$ in $N$, and these paths each use one secondary arc. 
See Figure~\ref{fig:mwact_to_tmnc}.
Furthermore (by virtue of the fact that $w(y, u_x) \neq 1$ for any $x,y \in X$, and therefore a tail node is never added above $spec\_u_x\_left$ or $spec\_u_x\_right$), every path in $N$ has at most one secondary arc.

This completes the construction of the species network $N$, and our problem instance.
\mjn{Observe that $N$ is time-consistent, since each time we insert a new secondary arc, its two endpoints are below every other previously inserted node.}
We now  show that 
 $(T,X,w)$ has an incomparable assignment of cost at most $k$ if and only if $(D, l)$ is $N$-consistent using at most $2k$ transfers.

\begin{figure*}[!h]
\begin{center}
\includegraphics[width= 0.5 \textwidth]{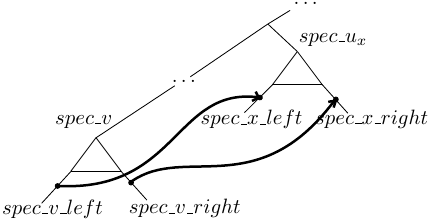}
\caption{Part of the species network $N$ constructed in the reduction from \textsc{ACT} to \textsc{NC}. For each $v \in T$, $spec\_v$  is the least common ancestor in $N$ of $spec\_v\_left$ and $spec\_v\_right$. If $w(x,v) = 1$ and $w(x, u_x)=0$, then secondary arcs (the thick lines) are added from an ancestor of $spec\_v\_left$ to an ancestor of $spec\_x\_left = spec\_u_x\_left$, and from an ancestor of $spec\_v\_right$ to an ancestor of $spec\_x\_right = spec\_u_x\_right$.
Thus, there are paths from $spec\_u_x$ to each of $spec\_x\_left$ and $spec\_x\_right$ using $0$ transfers in total, and paths from $spec\_v$ to each of $spec\_x\_left$ and $spec\_x\_right$ using $2$ transfers in total.
}
\label{fig:mwact_to_tmnc}
\end{center}
\end{figure*}

 First suppose that $(D, l)$ is $N$-consistent using at most $2k$ transfers.
We will first show the following claim. 
In this claim and its proof, we use the terms 'ancestor' and 'descendant' to exclusively refer to ancestors or descendants with respect to the distinguished base tree $T_0(N)$:

\begin{nclaim}\label{claim:0-or-2}
For $x \in X$, suppose $u \in V(N)$ is such that there exist paths from $u$ to $spec\_x\_left$ and from $u$ to $spec\_x\_right$, using at most $k_x$ secondary arcs in total.
If $k_x = 0$ then $u$ is an ancestor of $spec\_u_x$, and otherwise $u$ is an ancestor of 
some $spec\_v$ such that $w(x,v)\leq1$.
Moreover, if $u$ is not an ancestor of $spec\_u_x$ then $k_x = 2$.
\end{nclaim}
\begin{proof}
 First, recall that $spec\_u_x$ is the least common ancestor of $spec\_x\_left$ and \linebreak $spec\_x\_right$ in $T_0(N)$. Since $k_x = 0$ implies that $u$ is an ancestor of both  $spec\_x\_left$ and  $spec\_x\_right$, we have that if $k_x = 0$ then $u$ is an ancestor of $spec\_u_x$.
 
 Since there is a path from $u$ to $spec\_x\_left$, $u$ must be an ancestor of $spec\_v\_left$ for some $v$ such that $w(v,x) \leq 1$ (such nodes are the only ones that have a path to $spec\_x\_left$, either using exclusively principal arcs or a using a  single secondary arc).
 Similarly, $u$ must be an ancestor of $spec\_v'\_right$ for some $v'$ such that $w(v',x) \leq 1$.
 If $v = v'$ then $u$ is an ancestor of both  $spec\_v\_left$ and  $spec\_v\_right$ and is therefore an ancestor of $spec\_v$, as required. 
 So assume that $v \neq v'$. If $u$ is an ancestor of $spec\_v$ or $spec\_v'$ then we are done, and otherwise $u$ must be a descendant of both $spec\_v$ and $spec\_v'$ (since it is an ancestor of descendants of both of these). But this implies that $v$ and $v'$ are comparable, a contradiction as $w(x,v),w(x,v') < \infty$.
 
 Finally, we observe that $k_x < 1$ only if $u$ is an ancestor of at least one of  $spec\_x\_left$ and  $spec\_x\_right$. Therefore if $k_x < 2$ and $u$ is not an ancestor of $spec\_u_x$, it is a descendant of $spec\_u_x$. But this again implies  a contradiction as $u$ is an ancestor of some $spec\_v$ with $w(x,v)=1$, which would then be a descendant of $spec\_u_x$.
\end{proof}

Now consider the binary refinement $(D', l')$ of $(D, l)$ that is $N$-consistent using at most $2k$ transfers.
Thus there exists $\alpha$ such that $(D', \alpha)$ is a reconciled gene tree with respect to $N$.
Note that by construction of $(D, l)$, there is a rooted subtree in $D'$ whose leaves are the set \mj{of duplication nodes} $\{gene\_x: x \in X\}$ and whose internal nodes are all speciation nodes according to $l'$.
For each $x \in X$, there are paths in $D'$ from $gene\_x$ to $x\_left$ and to $x\_right$, and so there are paths in $N$ from $\alpha_\last(gene\_x)$ to $\sigma(x\_left) = spec\_x\_left$ and to $\sigma(x\_right)= spec\_x\_right$.
It follows from Claim~\ref{claim:0-or-2} that $\alpha_\last(gene\_x)$ is an ancestor of $v$ for some $v \in V(T)$ such that $w(x,v) \leq 1$.
By construction of $N$, there are no paths to such a $v$ using a secondary arc, and therefore as all ancestors of $x$ in $D'$ are speciation nodes, $\{\alpha_\last(gene\_x): x \in X\}$ must form the leaves of a subtree in $T$.
It follows that  $\alpha_\last(gene\_x)$ and $\alpha_\last(gene\_y)$ are incomparable for any $x \neq y \in X$.

Now  we can define $f:X \rightarrow V(T)$ as follows.
For each $x \in X$, let $f(x) = u_x$ if $\alpha_\last(gene\_x)$ is an ancestor of $spec\_u_x$,
and otherwise let $f(x)$ be a $v \in V(T)$ such that $w(x,v) \leq 1$ and $\alpha_\last(gene\_x)$ is an ancestor of $spec\_u_x$.
As their ancestors $\alpha_\last(gene\_x)$ and $\alpha_\last(gene\_y)$ are incomparable, it follows that $f(x)$ and $f(y)$ are also incomparable, for any $x \neq y \in X$.
Furthermore, by Claim~\ref{claim:0-or-2} we have that either $\alpha_\last(gene\_x)$ is an ancestor of $spec\_u_x$, or the paths from $\alpha_\last(gene\_x)$ to  $\sigma(x\_left)$ and to $\sigma(x\_right)$ use  $2$ secondary arcs.
Therefore the number of transfer arcs used by $\alpha$ is  $2$ for every $x \in X$ with $w(x,f(x)) = 1$.
Thus $2k \geq 2\sum_{x \in X}w(x,f(x))$, and so $f$ is an incomparable assignment with $\sum_{x \in X}w(x,f(x)) \leq k$, as required.

Now suppose that $(T,X,w)$ has an incomparable assignment  $f:X \rightarrow V(T)$ such that $\sum_{x \in X}w(x,f(x)) \leq k$.
We will show that $(D, l)$ has a binary refinement $(D', l')$ that is $N$-reconcilable using at most $2k$ transfers.
In particular, we will show that there is a reconciliation $\alpha$ such that $\alpha_\last(gene\_x) = spec\_f(x)$ for all $x \in X$.

Observe first that as $f$ is an incomparable assignment, there exists a subtree $T'$ of $T$ whose leaves are $\{f(x): x \in X\}$.
By refining the root $r$ of $D$ into a subtree isomorphic to $T'$, we get a refinement $(D', l')$ such that $D'$ with the leaves $\{x\_left,x\_right: x \in X\}$ removed has a reconciliation with $N$ using $0$ transfers. 
Furthermore this reconciliation $\alpha$ is such that  $\alpha_\last(gene\_x) = spec\_f(x)$ for all $x \in X$.
It remains to show how to extend $\alpha$ to the leaves $\{x\_left,x\_right: x \in X\}$ of $D'$.

For each $x \in X$, let $P_{x\_left}$ be a path in $N$ from $spec\_f(x)$ to $spec\_x\_left$ using a minimum number of secondary arcs.
By construction, this path uses $0$ secondary arcs if $w(x,f(x)) = 0$, and at most $1$ secondary arc if $w(x,f(x)) = 1$.
Similarly, let $P_{x\_right}$ be a path in $N$ from $spec\_f(x)$ to $spec\_x\_right$ using a minimum number of secondary arcs.
Then for each $x \in X$, we let $\alpha(x\_left) = P_{x\_left}$ and $\alpha(gene\_x\_right) = P_{x\_right}$.
It can be seen that $(D', \alpha)$ is a valid reconciliation with respect to $N$ that agrees with $(D', l')$.
Furthermore, $\alpha$ uses $2$ transfers for each $x \in X$ such that $w(x,f(x)) = 1$, and no others.
Therefore $D'$ is reconcilable using at most $\sum_{x \in X}2w(x,f(x)) \leq 2k$ transfers, as required. 
\end{proof}


\noindent
\textbf{Theorem \ref{thm:dp-algo-is-ok}.} 
\emph{Algorithm $minTransferCost$ is correct.  Moreover, it runs in time \linebreak 
$O(2^{k}k!k|V(D)||V(N)|^4)$ \ml{and space $O(|V(N)||V(D)| + |V(N)|^2)$}.}

\begin{proof}
We prove the following statement by induction:
for each $g \in V(D)$ and $s \in V(N)$, the algorithm finds the minimum number of required transfers for a reconciliation between 
the subtree $D_g$ and $N$ such that $g$ is mapped to $s$.
If $g$ is a leaf of $D$, the statement is easy to see, so suppose $g \in I(D)$.
Let $(\hat{D}_g, \alpha)$ be an optimal solution for $D_g, s$ and $N$, i.e. $\hat{D}_g$ is a binary refinement of $D_g$, $\alpha$ is 
a reconciliation between $\hat{D}_g$ and $N$ such that $\alpha_\last(g) = s$, 
and the pair $(\hat{D}_g, \alpha)$ minimizes the number $t$ of required transfers.
If $g$ is binary, then $g_l$ and $g_r$ are children of $g$ in both $D_g$ and $\hat{D}_g$.
Let $s_1 = \alpha_\last(g_l)$ and $s_2 = \alpha_\last(g_r)$.
It is clear that $\alpha$ restricted to $\hat{D}_{g_l}$\footnote{By the restriction $\alpha'$ of $\alpha$ to $\hat{D}_{g_l}$, we mean 
$\alpha'(v) = \alpha(v)$ for all strict descendants $v$ of $g_l$, and
\mj{$\alpha'(g_l) = (\alpha_\last(g_l))$}} yields a reconciliation of 
$\hat{D}_{g_l}$ using $f(g_l, s_1)$ transfers, since if there was a better refinement of $D_{g_l}$ admitting a better reconciliation
with $g_l$ mapped to $s_1$, then we could include this subsolution in $(\hat{D}, \alpha)$ and obtain a lower transfer cost.
The same argument holds for $g_r$ and $f(g_r, s_2)$.
We thus need to show that the algorithm will, at some point, consider the scenario of mapping 
$g_l$ with $s_1$ and $g_r$ with $s_2$.
If $l(g) = \S$, two cases may occur, according to Definition~\ref{def:DTLrecLGTNetwork}: 
(1) $e(\alpha_{\last}(g)) = \S$, in which case $\alpha_1(g_l) = s_l$ and $\alpha_1(g_r) = s_r$ (or vice-versa, w.l.o.g.).
This implies $s_1 \in P(s_l)$ and $s_2 \in P(s_r)$, and this scenario is tested on line~\ref{line:the-s-case} of $reconcileLBR$;
(2) $e(\alpha_{\last}(g)) = \T$, in which case $(s, s') $
is a transfer-arc, say $s' = s_r$ without loss of generality.
Then 
$\alpha_1(g_l) \in \{s, s_l\}$ and $\alpha_1(g_r) = s_r$ (or
vice-versa, w.l.o.g.), which imply $s_1 \in P(s)$ and $s_2 \in P(s_r)$.  This is tested by line~\ref{line:the-t-case} of $reconcileLBR$.
If $l(g) = \D$, we have $\alpha_1(g_l) = \alpha_1(g_r) = s$ and thus it is only required that $\alpha_\last(g_l) \in P(s)$ and $\alpha_\last(g_r) \in P(s)$, 
which is tested on line~\ref{line:the-d-case}.  
Therefore, the desired scenario of mapping $g_l$ to $s_1$ and $g_r$ to $s_2$ is considered.

One can also observe that no invalid mappings of $g_l$ and $g_r$ are considered by the algorithm
(if $l(g) = \S$, we test only the $s_1$ and $s_2$ that allow $e(\alpha_{\last}(g)) \in \{\S, \T\}$, 
and similarly for $l(g) = \D$).
The fact that the computed value $f'(g, s)$ (and hence $f(g, s)$) is minimum follows 
from the induction hypothesis on $g_l$ and $g_r$.

Suppose instead that $g$ has children $g_1, \ldots, g_k$, $k \geq 3$.
For a fixed $(D', l') \in \mathcal{B}(g)$, by the induction hypothesis we have that
$f(g_i, s')$ is correct for every $i \in [k]$ and $s' \in V(N)$. 
Using the argumentation for the binary case, it follows that after calling $reconcileLBR$,
we have correctly computed the minimum number of transfers for 
the tree obtained from $D_g$ after replacing $g$ by its local binary refinement $D'$.
The connected subtree $B_g$ of $\hat{D}$ induced by $g, g_1, \ldots, g_k$ is in $\mathcal{B}(g)$, 
and hence $minTransferCost$ will find $f(g, s)$ correctly when trying $D' = B_g$.
This concludes the proof, since the \ml{time and space} complexity of the algorithm was argued in the main text.
\end{proof}


\noindent
\textbf{Lemma \ref{lem:equiv-ds-tree-unknown}.} 
\emph{Let $R$ be a relation graph and $S$ be a species tree.
Then $R$ is \ml{$S$-base-consistent} (using $k$ transfers) if and only if 
there exists a least-resolved $DS$-tree $(D, l)$ that displays $R$ and a binary refinement $(D', l')$ of $(D, l)$ such that
$(D', l')$ is \ml{$S$-base-reconcilable} (using $k$ transfers).}

\begin{proof}
($\Rightarrow$) Assume that $R$ is \ml{$S$-base-consistent} using $k$ transfers.  Then there exists an LGT network $N$ such that $T_0(N) = S$ and $R$ is $N$-consistent
using $k$ transfers.  Then by Lemma~\ref{lem:equiv-ds-tree}, there is a $DS$-tree $(D, l)$ and a binary refinement 
$(D', l')$ such that $(D', l')$ is $N$-reconcilable using $k$ transfers.  Thus by definition, 
\mj{$(D', l')$}
is \ml{$S$-base-reconcilable} using $k$ transfers.

($\Leftarrow$) Assume that there is a DS-tree $(D, l)$ that displays $R$ and a binary refinement $(D', l')$ of $(D, l)$ such that $(D', l')$ is \ml{$S$-base-reconcilable} using $k$ transfers.  Then there is an LGT network $N$ such that $T_0(N) = S$ and $(D', l')$ is $N$-reconcilable using
$k$ transfers.  Again, by Lemma~\ref{lem:equiv-ds-tree}, $R$ is $N$-consistent using $k$ transfers.
So $R$ is also \ml{$S$-base-consistent} using $k$ transfers.
\end{proof}

\noindent
\textbf{Lemma \ref{lem:all-ds-trees-reconcilable}.} 
\emph{Let $(D, l)$ be a binary $DS$-tree and let $N := N(D)$ be the species network obtained from $S$ after
applying Algorithm~\ref{algo:everything-consistent}. 
Then $(D, l)$ is $N$-reconcilable.}

\begin{proof}
We show that for any $v \in I(D)$, the subtree
$(D_v, l)$ is $N$-reconcilable (where here, we slightly abuse notation by using $l$ to label $D_v$).  Moreover, we show that if $v$ is not a leaf and $s_i, s_j \in L(S)$ are distinct, then there is a reconciliation $(D_v, \alpha)$ with respect to $N$ 
such that $\alpha(v) = (\sdon{i}{j}{d(v)})$ and $e(\alpha_\last(v)) = \T$
(here, and for the rest of the proof, $d(v)$ refers to the depth of $v$ in $D$, and \emph{not}
its depth in $D_v$).
We use induction on the height $h(D_v)$.
First note that if $h(D_v) = 0$, then the statement is trivially true.

As an additional base case, suppose that $h(D_v) = 1$ and fix some $\sdon{i}{j}{d(v)}$, with $i \neq j$.  Then both children $v_l$ and $v_r$ of $v$ are leaves.
Let $s_p = \sigma(v_l)$ and $s_q = \sigma(v_r)$ for some $p, q \in [m]$.
Note that $p = q$ is possible.

We find two paths $P_1$ and $P_2$ that correspond to $\alpha(v_l)$ and $\alpha(v_r)$.
We first claim that in $N$, there exists a directed  path $P_1 = (\sdon{i}{j}{d(v)} = x_1, x_2, \ldots, x_{k_1} = s_p)$ 
such that $x_2 = \srec{j}{i}{d(v)}$ (i.e. $P_1$ starts with the $(\sdon{i}{j}{d(v)}, \srec{j}{i}{d(v)})$ arc).
Observe that  there exists a directed path $P_1'$ from 
$\srec{j}{i}{d(v)}$ to $s_p$.  Indeed, if $s_j = s_p$, then $\srec{j}{i}{d(v)} = \srec{p}{i}{d(v)}$ 
is an ancestor of $s_p$ and $P_1'$ obviously exists.
Otherwise, $P_1'$ starts from $\srec{j}{i}{d(v)}$, goes to its descendant 
$\sdon{j}{p}{d(v) + 1}$, takes the $(\sdon{j}{p}{d(v)+1}, \srec{p}{j}{d(v)+1})$ arc and then goes to $s_p$
(observe that $\sdon{j}{p}{d(v) + 1}$ does exist, since the first loop of the algorithm creating $N$
takes $c$ from $1$ to $h(D) + 1$, and $d(v) \leq h(D)$).  Since $P_1'$ exists and $(\sdon{i}{j}{d(v)}, \srec{j}{i}{d(v)})$ is an arc of $N$, the $P_1$ path exists.

By the same arguments, there is a path $P_2 = (\sdon{i}{j}{d(v)} = y_1, y_2, \ldots, y_{k_2} = s_q)$.   

Now, the existence of $P_1$ and $P_2$ imply that we can make $v$ a transfer node.
More precisely, we let 
\begin{align*}
\alpha(v) &= (\sdon{i}{j}{d(v)}) \\
\alpha(v_l) &= (x_2, x_3, \ldots, x_{k_1}= s_p) \\
\alpha(v_r) &= (\sdon{i}{j}{d(v)}, y_2, y_3, \ldots, y_{k_2} = s_q)
\end{align*}
Set $e(\alpha_\last(v)) = \T$ 
and $e(v_l, k) \in \{\SL, \TL, \emptyset\}$ for $k \in [|\alpha(v_l)| - 1]$ depending on what type of arc
$x_kx_{k + 1}$ is, then do the same for each $e(v_r, k)$ and $k \in [|\alpha(v_r)| - 1]$.  
We have $\alpha_\last(v) = \sdon{i}{j}{d(v)}$, $\alpha_1(v_l) = \srec{i}{j}{d(v)}$ and
$\alpha_1(v_r) = \sdon{i}{j}{d(v)}$, and since $(\sdon{i}{j}{d(v)}, \srec{j}{i}{d(v)}) \in E_s(N)$, condition a.4 of Definition~\ref{def:DTLrecLGTNetwork}
is satisfied, and so $\alpha$ is a reconciliation in which $e(\alpha_\last(v)) = \T$.  
This proves the base case.

Let $v \in V(D)$ such that $h(D_v) > 1$, and assume now by induction that the claim holds 
for any internal node $v'$ such that $D_{v'}$ has height smaller than $h(D_v)$.
Let $v_l, v_r$ be the children of $v$.
At least one of $v_l, v_r$ must be an internal node, say $v_r$ without loss of generality.
Suppose first that $v_l$ is a leaf.
As before, in $N$ there is a path $P_1 = (x_1, x_2, \ldots, x_{k_1})$ 
starting with the $(x_1, x_2) = (\sdon{i}{j}{d(v)}, \srec{j}{i}{d(v)})$ arc
and that goes to $x_{k_1} = \sigma(v_l)$.  As for $v_r$, by induction $D_{v_r}$ is $N$-reconcilable by some reconciliation $(D_{v_r}, \alpha')$ such that
$\alpha'(v_r) = (\sdon{i}{j}{d(v) + 1})$.  
Now, in $N$ there is a path $P_2 = (\sdon{i}{j}{d(v)} = y_1, y_2, \ldots, y_{k_2} = \sdon{i}{j}{d(v) + 1})$ 
from $\sdon{i}{j}{d(v)}$  
to $\sdon{i}{j}{d(v) + 1}$ in which each arc is in $E_p(N)$.
We can obtain the desired reconciliation $\alpha$ from $\alpha'$ in the following manner.
First let $\alpha(v) = (\sdon{i}{j}{d(v)})$ and $\alpha(v_l) = (x_2, x_3, \ldots, x_{k_1})$.
For every strict descendant $v_r'$ of $v_r$, let $\alpha(v_r') = \alpha'(v_r)$,
and finally let $\alpha(v_r) = (\sdon{i}{j}{d(v) + 1} = y_1, y_2, y_3, \ldots, y_{k_2} = \sdon{i}{j}{d(v) + 1})$.
As in the base case, we can set $e(\alpha_\last(v)) = \T$ and satisfy condition a.4 of Definition~\ref{def:DTLrecLGTNetwork}.  We set $e(v_r, k) \in \{\SL, \TL, \emptyset\}$ accordingly for every $k \in [|\alpha(v_r)| - 1]$
(depending on what type of arc $x_kx_{k+1}$ is)
and set $e(\alpha_\last(v_r)) = e(\alpha'_\last(v_r))$.
Finally we set $e(\alpha_k(v'_r))  = e(\alpha'_k(v'_r))$ for every strict descendant $v'_r$ of $v_r$
and every $k \in [|\alpha(v'_r)|]$.  We have that $\alpha(v), \alpha(v_l)$ and $\alpha(v_r)$ satisfy Definition~\ref{def:DTLrecLGTNetwork}, 
$e(\alpha_\last(v_r)) = e(\alpha'_\last(v_r))$
and every other gene-species mapping and event is unchanged from $\alpha'$.
It follows that $\alpha$ is a reconciliation.
Since $e(\alpha_\last(v)) = \T$, the claim is proved for this case.

If instead both $v_l, v_r \in I(D)$, then by induction, $D_{v_l}$
is $N$-reconcilable with 
 reconciliation $\alpha^l$
such that $\alpha^l(v_l) = (\sdon{j}{i}{d(v) + 1})$ (notice the use of $j \rightarrow i$ and not $i \rightarrow j$).
Moreover, $D_{v_r}$ is $N$-reconcilable with reconciliation $\alpha^r$ such that
$\alpha^r(v_r) = (\sdon{i}{j}{d(v) + 1})$.
In $N$, there is a path $P_1 = (x_1, x_2, \ldots, x_{k_1})$ 
starting with the $(x_1, x_2) = (\sdon{i}{j}{d(v)}, \srec{j}{i}{d(v)})$ arc
that goes to $x_{k_1} = \sdon{j}{i}{d(v) + 1}$.  
There is also a path $P_2 = (y_1, y_2, \ldots, y_{k_2})$
from $y_1 = \sdon{i}{j}{d(v)}$ to $y_{k_2} = \sdon{i}{j}{d(v) + 1}$ that uses only
arcs from $E_p(N)$.
Thus as before, we can make $v$ a transfer node.  That is we set $\alpha(v) = (\sdon{i}{j}{d(v)})$
and $e(\alpha_\last(v)) = \T$, $\alpha(v_l) = (x_2, \ldots, x_{k_1} = \sdon{j}{i}{d(v) + 1})$ and 
$\alpha(v_r) = (\sdon{i}{j}{d(v)} = y_1, y_2, \ldots, y_{k_2} = \sdon{i}{j}{d(v) + 1})$.
We set $e(v_l, k), e(v_r, k') \in \{\SL, \TL, \emptyset\}$ accordingly 
for every $k \in [|\alpha(v_l)| - 1], k' \in [|\alpha(v_r)| - 1]$, 
set $e(\alpha_\last(v_l)) = e(\alpha^l_\last(v_l)), e(\alpha_\last(v_r)) = e(\alpha^r_\last(v_r))$, and keep every other gene-species mapping and event from $\alpha^l$ and $\alpha^r$ unchanged.
In this manner $\alpha(v)$ satisfies Definition~\ref{def:DTLrecLGTNetwork},
and $\alpha$ is a reconciliation.
Again since $e(\alpha_\last(v)) = \T$, the claim is proved.
\end{proof}

\noindent
\textbf{Theorem \ref{thm:sconsistent-iff-dstree}.} 
\emph{A relation graph $R$ is \ml{$S$-base-consistent} if and only if
there exists a $DS$-tree $(D, l)$ such that $R(D, l) = R$.}

\begin{proof}
If there is no $DS$-tree $(D, l)$ such that $R(D, l) = R$, then by Lemma~\ref{lem:equiv-ds-tree}
there exists no species network $N$ with which $R$ is consistent, and thus $R$ cannot be
\ml{$S$-base-consistent}.
Conversely, let $(D' l')$ be a $DS$-tree such that $R(D', l') = R$, and let $(D, l)$ be a binary refinement
of $(D', l')$ (recalling that $R(D, l) = R(D', l') = R$).
Then by Lemma~\ref{lem:all-ds-trees-reconcilable}, $(D, l)$ is $N(D)$-reconcilable, where
the network $N(D)$ is the one constructed from $S$ by the algorithm described above.
By Lemma~\ref{lem:equiv-ds-tree}, $R$ is $N(D)$-consistent
and thus $R$ is also \ml{$S$-base-consistent}.
\end{proof}

\subsection*{Proof of Theorem~\ref{thm:hard-unknown-highways}: NP-hardness of minimizing transfers with unknown transfer highways}

The formal problem that we show NP-hard here in the following.

\medskip
\noindent \textsc{Transfer Minimization Species Tree Consistency  (TMSTC):}\\
\noindent {\bf Input}: A relation graph $R$, a species tree $S$, an integer $k$.\\
\noindent {\bf Question}: Is $R$ \ml{$S$-base-consistent} using at most $k$ transfers? \\

We reduce the feedback arc set problem to TMSTC. 

\medskip

\noindent \textbf{Feedback Arc Set (FAS):}\\
\noindent {\bf Input}: A directed graph $H = (V, A)$ and an integer $k$.\\
\noindent {\bf Question}: Does there exist a \emph{feedback arc set} of size at most $k$, 
i.e. a set of arcs $A' \subseteq A$ of size at most $k$ such that 
$H' = (V, A \setminus A')$ contains no directed cycle?\\

Given a FAS instance $H = (V, A)$, we construct a DS-tree $(D, l)$ and a species tree $S$ 
such that $H$ admits a feedback arc set of size at most $k$ if and only if 
$R(D, l)$ is \ml{$S$-base-consistent} using at most $K = 2|A| + k$ transfers.

A \emph{caterpillar} is a rooted binary tree in which every internal node has exactly one child that is a 
leaf, except for one node that has two leaf children.
We denote a caterpillar on leafset $x_1, x_2, \ldots, x_n$ by $(x_1|x_2|\ldots|x_n)$, where 
the $x_i$ nodes are ordered by depth in non-decreasing order (thus $x_1$ is the leaf child of the root).
A \emph{subtree caterpillar} is a rooted binary tree obtained by replacing some leaves 
of a caterpillar by rooted subtrees.
If each $x_i$ is replaced by a subtree $X_i$, we denote this by $(X_1|X_2|\ldots|X_n)$.  
If some $X_i$ is a leaf $x_i$ (i.e. $X_i$ a tree with one vertex $x_i$), we may write
$(X_1|\ldots|X_{i - 1}|x_i|\ldots|X_n)$.

Given the FAS instance $H = (V, A)$, first order $V$ and $A$ arbitrarily, 
and denote $V = (v_1, v_2, \ldots, v_n)$ and $A = (a_1, a_2, \ldots, a_m)$.
The species tree $S$ has a corresponding subtree for each vertex of $V$ and each arc of $A$.  For each vertex $v_i \in V$, 
let $S_{v_i}$ be a caterpillar $(v_{i,1} | v_{i,2} | \ldots | v_{i,2K})$ with $2K$ leaves.
For each $j \in [2K]$, denote $z_{i, j} = p(v_{i,j})$ (noting that $z_{i, 2K - 1} = z_{i, 2K}$).
Then, for each arc $a \in A$, let $S_{a}$ be the binary tree on two leaves $p_{a}, q_{a}$. 
Then $S$ is the subtree-caterpillar 
$(S_{a_1} | S_{a_2} | \ldots | S_{a_m} | S_{v_1} | S_{v_2} | \ldots | S_{v_n} )$.
See Figure~\ref{fig:fas-reduction}.

\begin{figure*}[!b]
\begin{center}
\includegraphics[width= 1.00 \textwidth]{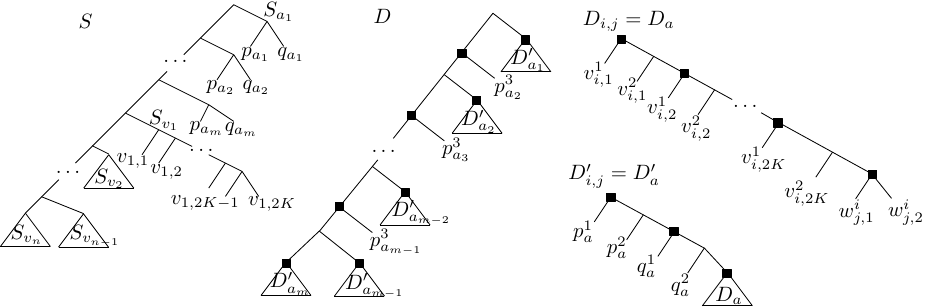}
\caption{The $S$ and $D$ trees constructed for our reduction.  Duplication nodes appear as squares, and the absence of a square indicates speciation.}\label{fig:fas-reduction}
\end{center}
\end{figure*}

The DS-tree $(D, l)$ has one subtree for each arc of $A$.  For each $a = (v_i, v_j) \in A$, 
let $D_{a} = D_{i,j}$ be a caterpillar with $4K + 2$ leaves
such that 
\[
D_{i,j} = (v_{i,1}^1 | v_{i,1}^2 | v_{i,2}^1 | v_{i,2}^2 | \ldots | v_{i,2K}^1 | v_{i,2K}^2 | w_{j,1}^i | w_{j,2}^i)
\](we will interchangeably use the $D_a$ and $D_{i,j}$ notations whenever convenient). 
Here the indices of the leaf labels indicates the species containing them, 
i.e. for each $h \in [2K], \sigma(v_{i, h}^1) = \sigma(v_{i, h}^2) = v_{i,h}$, 
and $\sigma(w_{j,1}^i) = v_{j,1}, \sigma(w_{j,2}^i) = v_{j,2}$.
Thus all the leaves of $L(D_{i,j})$ are from the $S_{v_i}$ subtree, with the exception of 
$w_{j,1}^i$ and $w_{j,2}^i$ at the bottom.
For each $h \in [2K]$, the parent of $v_{i, h}^1$ is labeled by $\D$ whereas the parent of $v_{i, h}^2$
is labeled by $\S$.   
The parent of $w_{j,1}^i$ and $w_{j,2}^i$ is labeled by $\D$.
We define another tree $D'_a = D'_{i,j} = (p_a^1 | p_a^2 | q_a^1 | q_a^2 | D_{i,j})$.
The parents of $p_a^1$ and $q_a^1$ are labeled $\D$,
whereas the parents of $p_a^2$ and $q_a^2$ are labeled $\S$
(here $\sigma(p_a^1) = \sigma(p_a^2) = p_a$ and $\sigma(q_a^1) = \sigma(q_a^2) = q_a$).

Finally, we let 
\[D = (D'_{a_1} | p^3_{a_2} | D'_{a_2} | p^3_{a_3} | D'_{a_3} | \ldots | p^3_{a_{m - 2}} | D'_{a_{m - 2}} | p^3_{a_{m - 1}} | D'_{a_{m - 1}} | D'_{a_m})\]

where each $p_{a_i}^3$ is a new leaf with $\sigma(p_{a_i}^3) = p_{a_i}$.  The purpose of the $p_{a_i}^3$ is to enforce a binary DS-tree.
The root is a speciation, and the main path of $D$ alternates labelings, i.e. for each $1 < i < [m]$, 
the parent of $p^3_{a_i}$ is labeled $\D$ and the parent of $r(D'_{a_i})$ is labeled $\S$.
The parent of $r(D'_{a_m})$ is labeled $\S$.

It is not hard to see that this construction can be carried out in polynomial time.
Note that $D$ is binary and is also a least-resolved $DS$-tree.
Thus by Lemma~\ref{lem:equiv-ds-tree-unknown}, $R(D, l)$ is \ml{$S$-base-consistent} using $K$ transfers if and only if 
$(D, l)$ is \ml{$S$-base-reconcilable} using $K$ transfers.

\begin{lemma}
If $H$ admits a feedback arc set $A' \subset A$ of size $k$, then 
$(D, l)$ is \ml{$S$-base-reconcilable} using at most $K = 2m + k$ transfers.
\end{lemma}

\begin{proof}
The intuition behind the proof is as follows.
Each $D_{i,j}$ subtree and its $\D, \S$ labeling could be part of a valid reconciliation with respect to $S$, if it were not for the
$w^i_{j, 1}$ and $w^i_{j, 2}$ leaves at the bottom, which prevent their ancestors to be speciations.  These need to be handled by either making the two edges incident to 
$w^i_{j, 1}$ and $w^i_{j, 2}$ a transfer to $v_{j, 1}$ and $v_{j, 2}$ respectively, or better, by making the edge 
above their common parent a transfer to some common ancestor of $v_{j, 1}$ and $v_{j, 2}$.  The latter option is preferred as it requires one less transfer, but it cannot be taken 
for every $D_{i, j}$ subtree because we will likely create time-inconsistencies. 
As it turns out, given a feedback arc set $A'$ of size $k$, we have a way of taking these `double-transfers' only $k$ times.
As mentioned before, this is similar to the proof in~\cite{THL2011}.  The difficulty here however, is to ensure that time-consistency is preserved and 
that the $\D, \S$ labeling can be preserved.

We first show how to add secondary arcs to $S$ in a time-consistent manner in order to obtain $N$, 
by making the time function $t$ explicit.  We will add more arcs than necessary, but this simplifies the exposition.
Let $s_1,  \ldots, s_{n + m - 1}$ be the vertices on the $r(S) - r(S_{v_n})$ path in $S$ (excluding $r(S_{v_n})$),
ordered by depth in increasing order.  Assign time slot $t(s_{\l}) = \l$ for each $\l \in [n + m - 1]$.
We then describe the transformation from $S$ to $N$ in three steps.

\noindent
\textbf{Step 1: transfer arcs from $q_{a_{\l}}$ to $S_{v_i}$.}
We process each arc $a_{\l} \in A$ for $\l = 1,2, \ldots, m$ in increasing order as such:
first let $(v_i, v_j) = a_{\l}$ (i.e. $v_i, v_j$ are the vertices of the $a_{\l}$ arc in $H$).
Assign time slot $\l + 1$ to the parent of nodes $p_{a_{\l}}$ and $q_{a_{\l}}$.
Then, subdivide $(q_{a_{\l}}, p(q_{a_{\l}}))$, creating a new node that we call $send\_q_{a_{\l}}\_to\_i$.  Next, subdivide $(p(r(S_{v_i})), r(S_{v_i}))$, creating a new node that we call $recv\_i\_from\_q_{a_\l}$.  After that, we add the secondary arc 
$(send\_q_{a_{\l}}\_to\_i, recv\_i\_from\_q_{a_\l})$.  
See Figure~\ref{fig:fas-reconcil}(1) for an illustration.  Assign the time slot $m + n + \l$ to the two newly created nodes.

Note that this process is repeated for each arc $a_{\l}$ in order. Therefore, $p(r(S_{v_i}))$ may change during the process as new secondary arcs are inserted.
In the end, there is exactly one outbound transfer node inserted above each $q_{a_{\l}}$, and $|N^+(v_i)|$ inbound transfer nodes inserted above each $r(S_{v_i})$, where 
$N^+(v_i)$ is the set of out-neighbors of $v_i$ in $H$.
One can check that no time inconsistency is created so far, since every time a node is inserted, it is added 
below every other internal node having a defined time slot so far, and it is assigned a higher time slot (since $m + n + \l$ is always the highest time slot so far, for each $\l \in [m]$).
Also note for later reference that, assuming $n \leq m$, $t(p(r(S_{v_i}))) \leq m + n + \l$ for some $\l \leq m$, and therefore $t(p(r(S_{v_i}))) \leq 3m$ after these operations.

\begin{figure*}[!h]
\begin{center}
\includegraphics[width= 1.00 \textwidth]{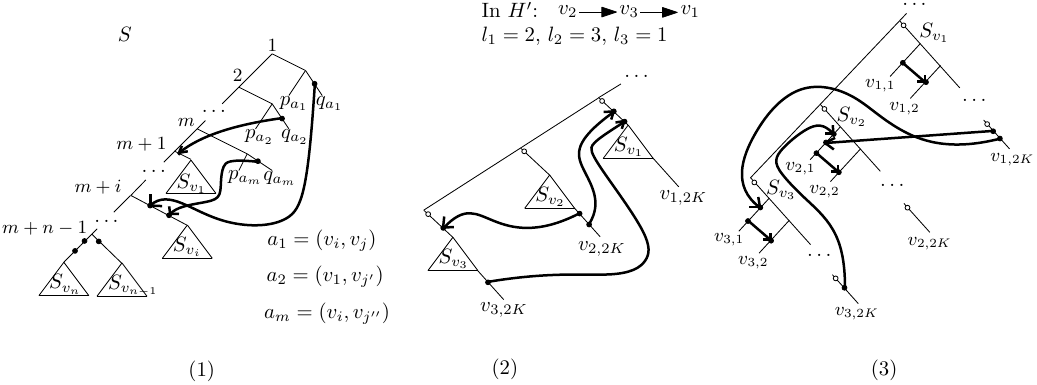}
\caption{An illustration of the modifications from $S$ to $N$. (1) We first add the transfers between the $S_{a_{\l}}$ subtrees to the $S_{v_i}$ subtrees.
For the purpose of the example, we have only illustrated the arcs $a_1 = (v_i, v_j), a_2 = (v_1, v_{j'}), a_m = (v_i, v_{j''})$ (the $j, j', j''$) indices are irrelevant for this step).  
Here the node added above $q_{a_1}$ would be named $send\_q_{a_1}\_to\_i$ and its endpoint is $recv\_i\_from\_q_{a_{1}}$. 
(2) We then add ``forward-transfers'', which are secondary arcs from the bottom of $S_{v_{l_i}}$ to the top of $S_{v_{l_j}}$, where $j > i$.  Here we illustrate this step on a small example
of $H'$, with the topological sorting $(v_2, v_3, v_1)$.  The white nodes indicate that other transfer nodes could be on the subpath due to the previous step.
(3) We finally allow transferring ``backwards'' from $v_{i, 2K}$ to $v_{j, 1}$, $j< i$, then from $v_{j, 1}$ to $v_{j, 2}$.}
\label{fig:fas-reconcil}
\end{center}
\end{figure*}

For what follows, let $H' = (V, A \setminus A')$.  Since $H'$ is a directed acyclic graph, it admits a
topological sort, i.e. an ordering $(v_{l_1}, v_{l_2}, \ldots, v_{l_n})$ of $V$ 
such that if $i < j$, then $(v_{l_j}, v_{l_i})$ is not an arc of $H'$ (in other words, there are no backwards arcs).
\man{We now add two new sets of arcs that are entirely based on the ordering $(v_{l_1}, \ldots, v_{l_n})$.}

\noindent
\textbf{Step 2: transfer arcs from $v_{l_i, 2K}$ to its successor subtrees.}
What we want to achieve in this step is that for each $v_{l_i}$, we can transfer from the parent of $v_{l_i, 2K}$ to any subtree $S_{v_{l_h}}$ such that $h > i$.  
An example is provided in Figure~\ref{fig:fas-reconcil}(2).
Process each vertex $v_{l_i} \in V$ for $i = 1, 2, \ldots, n$ in increasing order as follows.
First we create the transfer nodes above $r(S_{v_{l_i}})$ that are destined to receive from the predecessors of $v_{l_i}$.
For each $j = 1, 2, \ldots, i - 1$ in order, add a node $recv\_l_i\_from\_l_j$ on the edge between $r(S_{v_{l_i}})$
and its parent, and assign the time slot 
\[t(recv\_l_i\_from\_l_j) = (4 + i)Km + j\]
Then, we create the nodes above $v_{l_i, 2K}$ that are destined to send to the successor subtrees of $v_{l_i}$.  For each $j = i + 1, i + 2, \ldots, n$ in increasing order, 
add a node $send\_l_i\_to\_l_j$ on the $(p(v_{l_i,2K}),v_{l_i, 2K})$ arc.  
For each such $j$, assign time slot 
\[t(send\_l_i\_to\_l_j) = (4 + j)Km + i\]

Then, for each $i, j \in [n]$ with $i < j$, add a transfer arc from 
$send\_l_i\_to\_l_j$ to \linebreak $recv\_l_j\_from\_l_i$.  
Note that this transfer arc satisfies our time consistency requirement since 
$t(send\_l_i\_to\_l_j) = (4 + j)Km + i = t(recv\_l_j\_from\_l_i)$.
Also note that for each arc $(v_{l_i}, v_{l_j})$ in $A \setminus A'$, 
there is a corresponding secondary arc from 
$send\_l_i\_to\_l_j$ to $recv\_l_j\_from\_l_i$.

We argue that $S$ is still time-consistent.  We know already that secondary arcs so far have the same timing, 
so we must show that 
(1) no node has a child with a greater time slot, and 
(2) there is a way to assign a time slot to the nodes $z_{i, 1}, \ldots, z_{i, 2K-1}$ within the $S_{v_i}$ trees. 
For (1), all the receiving and sending nodes 
inserted at the last step have 
a time slot greater than $3m$ and are inserted below the nodes that had a time slot assigned
at the previous step (which were assigned a time slot at most $3m$).  
Moreover, the $recv\_l_i\_from\_l_j$ nodes are inserted on the 
$p(r(S_{v_{l_i}}))r(S_{v_{l_i}})$ arc in increasing order of time, 
as well as the $send\_l_i\_to\_l_j$ nodes on the $(p(v_{l_i,2K}), v_{l_i, 2K})$ arc.
Hence no inconsistency is created within the $S_{v_i}$ trees.
For (2), note that for each $i \in [m]$, 
the nodes $z_{i, 1}, \ldots, z_{i,2K-1}$ of $S_{v_i}$ lying on the path between $recv\_l_i\_from\_l_{i - 1}$
(above $r(S_{v_i})$)
and $send\_l_i\_to\_l_{i + 1}$ (at the bottom of $S_{v_i}$) all have an available time slot between 
$(4 + i)Km + i - 1$ and $(4 + i + 1)Km + i$, since there are $2K - 1$
such nodes and there are $Km + 1$ available time slots.  
Therefore, we can assign a time to each $z_{i, h}$ so that time consistency holds.  Note that all internal nodes of $S$ have been assigned a time slot so far.

\noindent
\textbf{Step 3: escape route from \man{$v_{l_i, 2K}$ to $v_{l_j, 1}$, then to $v_{l_j, 2}$.}}
Again, process each vertex $v_{l_i}$ for $i = 1,2, \ldots, n$ in increasing order.
We make, for $j < i$, a ``last-resort escape route'' from $v_{l_i, 2K}$ to $v_{l_j, 1}$, 
followed by a transfer arc going from $v_{l_j, 1}$ to $v_{l_j, 2}$.  Taking these arcs in a reconciliation corresponds to taking ``backwards arcs'',
i.e.  that belong to $A'$.  
For that purpose, we add, on the arc between $v_{l_i, 2K}$ and its parent, $i - 1$ transfer nodes to send backwards.
Then on the arc between $v_{l_i, 1}$ and its parent, we add $n - i$ transfer nodes
to receive from the front.  This step is illustrated on Figure~\ref{fig:fas-reconcil}(3).

More precisely, 
for each $j = 1, 2, \ldots, i - 1$, add a node $backsend\_l_i\_to\_l_j$ on the edge between 
$v_{l_i, 2K}$ and its parent.  Assign a high time slot to this node, 
say for example $t(backsend\_l_i\_to\_l_j) = (Km)^{10} + i + j$.
Then for each $j = i + 1, i + 2, \ldots, n$, add a node $backrecv\_l_i\_from\_l_j$ on the 
edge between $v_{l_i, 1}$ and its parent.  Assign the time slot $t(backrecv\_l_i\_from\_l_j) = (Km)^{10} + i + j$.
Note that time consistency is still preserved by these node insertions.
Then for each $i, j \in [n]$ with $i > j$, add a secondary arc from 
$backsend\_l_i\_to\_l_j$ to $backrecv\_l_j\_from\_l_i$.
Again, these arcs are time-consistent since $t(backsend\_l_i\_to\_l_j) = (Km)^{10} + i + j = t(backrecv\_l_j\_from\_l_i)$.

To finish the network, for each $i \in [n]$, add a secondary arc $(send12\_i, recv12\_i)$ between 
the $(p(v_{i, 1}), v_{i, 1})$ arc and the $(p(v_{i, 2}), v_{i, 2})$ arc.
To preserve time-consistency, assign a large enough time slot, 
say $m^{100}$ to both newly created nodes. 
This finally concludes the construction.  Let us call the resulting network $N$.  

For the remainder, let $u, v \in V(N)$ and 
\mj{suppose}
that there is a path from $u$ to $v$ in $N$ that does not use a secondary arc.  We denote this path by $[u \isep v]$. 
We will also denote by $]u \isep v]$ the path $[u \isep v]$, but excluding $u$ from this path.

\noindent
\textbf{Reconciling $(D, l)$ with $N$.}
We are finally ready to show that $(D, l)$ is $N$-reconcilable using at most $K$
transfers.  We begin by showing how to reconcile $D_{i,j}$ for 
$a = (v_i, v_j) \in A$.  
For reasons that will become apparent later, the edge above $r(D_{i,j})$ will always be contain a transfer.
To be more precise, set $\alpha_1(p(v^1_{i,1})) = recv\_i\_from\_q_a$ with 
$e(p(v^1_{i,1}), 1) = \emptyset$ (setting it up to receive a transfer).  Then set  
$\alpha_\last(p(v^1_{i,1})) = z_{i,1}$ with $e(p(v^1_{i,1}), \last) = \D$.
Since there is a directed path from $recv\_i\_from\_q_a$ to $z_{i, 1}$ that uses 
no secondary arc of $N$, $\alpha(p(v^1_{i,1}))$ can be completed with the appropriate
$\SL$ events.
Set $\alpha(p(v^2_{i,1})) = (z_{i, 1})$ and 
for each $2 \leq h \leq 2K - 1$,
set
 $\alpha(p(v^1_{i, h})) = \alpha(p(v^2_{i, h})) = (z_{i, h})$
 (we will handle the case $h = 2K$ later).  Then 
set $e(\alpha_\last(p(v^1_{i,h}))) = \D$ and 
$e(\alpha_1(p(v^2_{i,h}))) = \S$.  Note that the assigned events are the same as in 
the $DS$ labeling $l$ of $D$, and that so far $\alpha$ satisfies Definition~\ref{def:DTLrecLGTNetwork}.
It is straightforward to set $\alpha(v^1_{i, h})$ and $\alpha(v^2_{i,h})$ appropriately.

\begin{figure*}[!h]
\begin{center}
\includegraphics[width= 1.00 \textwidth]{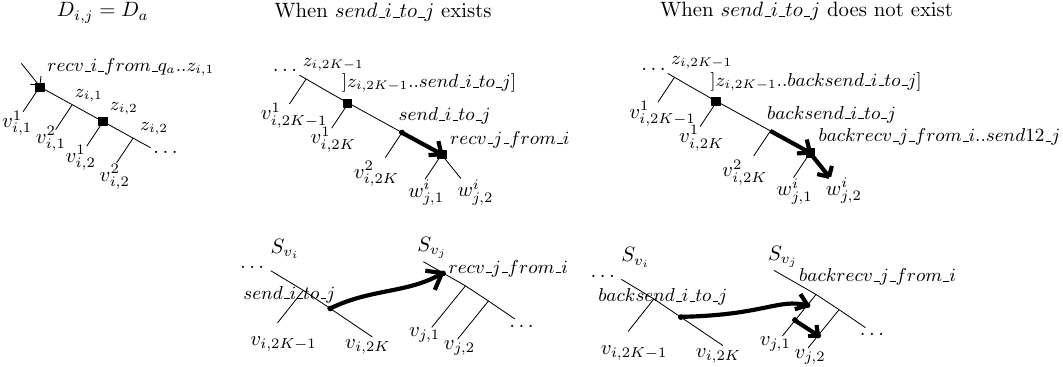}
\caption{Top left: how the $D_{i,j}$ subtree is reconciled from its root down to the parent of $v^1_{i, 2K - 1}$.  Top center and top right : the two possible reconciliations of $D_{i, j}$.  In the first case, we can handle the $w^i_j$ nodes using a single transfer above $S_{v_j}$.
In the second case, we must transfer on the arc leading to $v_{j, 1}$, then use another to get to $v_{j, 2}$.
Bottom: the transfer highways of $N$ used by both scenarios.}
\label{fig:fas-reconcile_G}
\end{center}
\end{figure*}

We now handle the nodes $p(v^1_{i,2K})$ and $p(v^2_{i,2K})$ (see Figure~\ref{fig:fas-reconcile_G} for an illustration).
First denote by $w$ the parent of both $w_{j,1}^i$ and $w_{j,2}^i$ in $D_{i,j}$.
Suppose that $a = (v_i, v_j)$ is not in $A'$.  
Recall the ordering $v_{l_1}, \ldots, v_{l_n}$ from above.
Then there are $i'$ and $j'$ such that $i = l_{i'}$ and $j = l_{j'}$, with $i' < j'$.  
Therefore $N$ has a secondary arc 
\[
(send\_l_{i'}\_to\_l_{j'}, recv\_l_{j'}\_from\_l_{i'}) = (send\_i\_to\_j, recv\_j\_from\_i)
\]
starting above $v_{i, 2K}$ and ending above $S_{v_j}$. 
We make the parent edge of $w$ borrow this transfer arc.
For that purpose, set \\
\man{$\alpha(p(v^1_{i, 2K})) =~]z_{i, 2K - 1} \isep send\_i\_to\_j]$}
and $\alpha(p(v^2_{i, 2K})) = (send\_i\_to\_j)$, \\
setting $e(p(v^1_{i, 2K}), \last) = \D$ and $e(p(v^2_{i, 2K}), \last) = \T$.  
For the child leaves, set $\alpha(v^1_{i, 2K}) = \alpha(v^2_{i, 2K}) = [send\_i\_to\_j \isep v_{i, 2K}]$.
Then we set $\alpha(w) = [recv\_j\_from\_i \isep z_{j, 1}]$
with $e(w, \last) = \D$.  It is straightforward to check that $\alpha(w_{j, 1}^i)$ 
and $\alpha(w_{j, 2}^i)$ can be set without requiring any additional transfer, 
since $z_{j,1}$ is an ancestor of both $v_{j, 1}$ and $v_{j, 2}$.

Now, suppose instead that 
$a = (v_i, v_j) \in A'$.  Then the transfer arc used in the previous case does not exist, 
since it is backwards with respect to our ordering.
In this case, we must use the last-resort route, namely the secondary arc\\ $(backsend\_i\_to\_j, backrecv\_j\_from\_i)$, 
then the $(send12\_j, recv12\_j)$ arc.
More precisely, set
\[
\alpha(p(v^1_{i, 2K})) =~]z_{i, 2K - 1} \isep backsend\_i\_to\_j]
\]
and 
\[
\alpha(p(v^2_{i, 2K})) = (backsend\_i\_to\_j)
\]  
with $e(p(v^1_{i, 2K}), \last) = \D$ and $e(p(v^2_{i, 2K}), \last) = \T$.  Then set $\alpha(v^1_{i, 2K}) = \alpha(v^2_{i, 2K}) = [backsend\_i\_to\_j \isep v_{i, 2K}]$.  
Then let $\alpha(w) = [backrecv\_j\_from\_i \isep send12\_j]$
with $e(w, \last) = \T$.  Set $\alpha(w_{j, 1}^i) = [send12\_j \isep v_{j, 1}]$ and 
$\alpha(w_{j, 2}^i) = [recv12\_j \isep v_{j,2}]$.  One can check that $\alpha$
satisfies Definition~\ref{def:DTLrecLGTNetwork} and in this case, $D_{i,j}$ requires two transfers.

It remains to reconcile the rest of $D$.
We exhibit $\alpha$ for the nodes of $D'_{i,j}$ that are not in $D_{i, j}$.
Denote $a = (v_i, v_j)$.
In $S$, denote $r_a = p(p_a) = p(q_a)$.  Set $\alpha(p(p^1_a)) = \alpha(p(p^2_a)) = (r_a)$,
and $e(p(p^1_a)) = \D, e(p(p^2_a)) = \S$ (we will adjust $\alpha(p(p^1_a))$ later).  
Then set $\alpha(p(q^1_a)) = [r_a \isep send\_q_a\_to\_i]$ with $e(p(q^1_a), \last) = \D$,
and $\alpha(p(q^2_a)) = (send\_q_a\_to\_i)$.  Recall that $p(v^1_{i, 1})$ is a child of $p(q^2_a)$
and that $\alpha_1(p(v^1_{i, 1})) = recv\_i\_from\_q_a$.  Thus by setting $e(p(q^1_a)) = \T$ 
we satisfy Definition~\ref{def:DTLrecLGTNetwork}.  It is clear that the $\alpha$ values for the leaves
$p^1_a, p^2_a, q^1_a$ and $q^2_a$ can be set without requiring any additional transfer.
We have now reconciled $D'_{i,j}$ such that $\alpha_\last(r(D'_{i,j})) = r_a$,
adding one transfer in the process.

What remains now are the nodes $g_1, g_2, \ldots, g_{\l}$, ordered by increasing depth, that lie on the path
between $r(D)$ and $r(D'_{a_m})$ (excluding the latter).
We claim that none of these nodes requires any transfer.
The node $g_{\l}$ is a speciation and has two children $r(D'_{a_{m - 1}})$ and $r(D'_{a_m})$: one mapped by $\alpha$ to species $r_{a_{m - 1}}$ and the other to $r_{a_m}$.  
Then we can set $\alpha(g_{\l}) = (lca_S(r_{a_{m - 1}}, r_{a_m}))$ and $e(g_{\l}, \last) = \S$, and adjust $\alpha(r(D'_{a_{m - 1}}))$ and $\alpha(r(D'_{a_{m }}))$ accordingly.  Now, $g_{\l - 1} = p(g_{\l})$ is a duplication whose
other child is $p^3_{m - 1}$, and thus it is safe to set $\alpha(g_{\l - 1}) = (lca_S(r_{a_{m - 1}}, r_{a_m}))$ as well
and set $e(g_{\l - 1}, \last) = \D$.  Since the $D_{a}$ subtrees are ordered in the same manner in $D$ as the $S_a$ subtrees in $S$, it is not hard to see inductively that for $i < \l - 1$,
if $l(g_{i}) = \S$, then $g_i$ has $r(D'_{a_h})$ as a child for some $h < m - 1$, which is mapped to 
$r_{a_h}$, and the other child is $g_{i + 1}$, 
mapped to $x := lca_S(r_{a_{h + 1}}, r_{a_{h + 2}})$.  Hence we can set $\alpha(g_i) = (x, r_{a_h})$
and adjust the $\alpha$ values of the two children of $g_i$ accordingly.  If $l(g_i) = \D$, we simply set
$\alpha(g_i) = \alpha(g_{i + 1})$.
We are done with the reconciliation $\alpha$ between $D$ and $N$.

To sum up, if $a \notin A'$, then $D'_a$ requires $2$ transfers, 
and if $a \in A'$, then $D'_a$ requires $3$ transfers, and $|A'| = k$.
Thus $K = 2m + k$ transfers are added in total.
\end{proof}

We now undertake the converse direction of the proof.
We will make use of the following well-known fact on reconciliations.

\begin{lemma}\label{lem:lcaspec}
Let $S$ be a species tree and let $N$ be an LGT network obtained by adding secondary arcs to $S$.
Let $(D, \alpha)$ be a reconciliation with respect to $N$.  Let $u \in I(D)$ such that $e(u, \last) = \S$ and let $v, w$ be two leaves descending from $u$ such that, 
for every node $z$ on the path between $u$ and $v$ or on the path between $u$ and $w$, $\alpha(z)$ contains no $\T$ or $\TL$ event.
Then $\alpha_{\last}(u) = lca_{S}(\sigma(v), \sigma(w))$.
\end{lemma}

\begin{proof}
First note that by the definition of a reconciliation, $\alpha_{\last}(u) = \S$ implies that $\alpha_{\last}(u)$ must exist in $S$, since only those nodes can be the tail of two principal arcs in $N$ (recall that this is required by speciation).

Assume without loss of generality that $v$ descends from $u_l$ and $w$ from $u_r$.  Let $P_v = (u = v_1, \ldots, v_a = v)$ be the path from $u$ to $v$ and $P_w = (u = w_1, \ldots, w_b = w)$ the path from $u$ to $w$.  By the definition of speciation, $\alpha_1(u_l)$ and $\alpha_1(u_r)$ are the two children of $\alpha_{\last}(u)$. 
Moreover, by appending the paths $\alpha(v_2), \ldots, \alpha(v_a)$ and eliminating possible repetitions due to duplications, we obtain a path $P'_v$ of $N$ that uses only principal arcs, starts at $\alpha_1(u_l)$ and ends at $v$.  Similarly, appending the paths $\alpha(w_2), \ldots, \alpha(w_b)$, we obtain a path $P'_w$ of $N$ that uses only principal arcs, starts at $\alpha_1(u_r)$ and ends at $w$.  
Because $\alpha_1(u_l)$ and $\alpha_1(u_r)$ are the children of $\alpha_{\last}(u)$ and $P'_v$ and $P'_w$ use only $E_p$ arcs, $P'_v$ and $P'_w$ are vertex-disjoint.
Thus $\alpha_{\last}(u)$ is a node of $N$ whose two children can start disjoint paths that lead to $v$ and $w$, respectively.
The only node of $N$ from which this is possible is $lca_S(\sigma(v), \sigma(w))$.
\end{proof}

\begin{lemma}
If $(D, l)$ is \ml{$S$-base-reconcilable} using at most $K = 2m + k$ transfers, 
then $H$ admits a feedback arc set $A' \subseteq A$ of size at most $k$.
\end{lemma}

\begin{proof}
Suppose that $(D, l)$ is \ml{$S$-base-reconcilable} using at most $K$ transfers, let  $N$ be the species network 
such that $T_0(N) = S$
and let $(D, \alpha)$ a reconciliation with respect to $N$ using $K$ transfers showing that $(D, l)$ is $N$-reconcilable.
We divide this proof into a series of claims.
Without loss of generality, we assume that the secondary arcs on $N$ are minimal, in the sense \mj{that}
every secondary arc of $N$ is used by $\alpha$.

\begin{nclaim}\label{claim:gprime-tl}
For every arc $a = (v_i, v_j) \in A$, in the $D'_{i,j}$ subtree, there is a 
node $x$ and an integer $h$ such that $e(x, h) \in \{\T, \TL\}$ and $x$ does not belong to $D_{i,j}$.
\end{nclaim}

\begin{proof}
Suppose for contradiction that the claim is false.
Denote $y_p := p(p^2_a)$ and $y_q = p(q^2_a)$.
Because there is no transfer, we have $e(y_p, \last) = e(y_q, \last) = \S$, by the orthology requirements of $(D, l)$.
By Lemma~\ref{lem:lcaspec}, $\alpha_{\last}(y_p) = lca_S(\sigma(p^2_a), \sigma(q^2_a)) = p(p_a) = p(q_a)$.
Now consider $\alpha_{\last}(y_q)$.  
By definition of speciation and by the absence of transfers in $\alpha(p(p_a^1))$ and $\alpha(y_q)$, $\alpha_{\last}(y_q)$ must be a strict descendant of $\alpha_{\last}(y_p) = p(q_a)$.
On the other hand, $\alpha_{\last}(y_q)$ is a strict ancestor of $q_a$ since $e(y_q, \last) = \S$.
Moreover, $\alpha_{\last}(y_q)$ is a node of $S$ \man{(this is because $e(y_q, \last) = \S$, and thus by definition, $\alpha_{\last}(y_q)$ must be a node whose two children are principal arcs)}.  We have reached a contradiction, since $S$ contains no node that is a strict descendant of $p(q_a)$ and a strict ancestor of $q_a$.
\end{proof}

\begin{nclaim}\label{claim:in_svi}
Let $(v_i, v_j) \in A$.  Then there is an internal node $x$ of $D_{i,j}$ such that
$\alpha_{\last}(x)$ is a node of $S_{v_i}$.
\end{nclaim}

\begin{proof}
Suppose that for every internal node $x$ of $D_{i,j}$, $\alpha_{\last}(x)$ is not a node of $S_{v_i}$. 
Let $h \in [2K]$ such that $h$ is odd.  
We show that there must be a transfer in some node of the path
between $v^2_{i,h}$ and $v^2_{i, h + 1}$ in $D_{i,j}$.
Let us assume that this is not the case.
We can thus assume that $e(p(v^2_{i,h}), \last) = \S$ \man{and
that $\alpha(v^2_{i,h})$ does not contain a $\TL$ event.}
It follows that $\alpha_{\last}(p(v^2_{i,h}))$ is an ancestor of $v_{i,h}$ which, by assumption, does not belong to $S_{v_i}$.
Since we further assume that there is no transfer in $\alpha(p(v_{i,h+1}^1)), \alpha(p(v_{i,h+1}^2))$ or $\alpha(v_{i,h+1}^2)$, by Lemma~\ref{lem:lcaspec}, we must have $\alpha_{\last}(p(v^2_{i, h}) = lca_S(\sigma(v_{i,h}), \sigma(v_{i, h+1}))$.
This node is in $S_{v_i}$, and we have reached a contradiction.
Therefore, some transfer must be present in some node of the $v_{i, h}^2 - v^2_{i,h+1}$ path.

This holds for every odd $h$, so $D_{i,j}$ has a least $K$ transfers.
But by the previous claim, $D'_{i,j}$ has at least one transfer that is not in $D_{i,j}$, so in total $D$ has strictly more 
than $K$ transfers, a contradiction.   
\end{proof}

\begin{nclaim}  \label{claim:gij-paths}
Let $(v_i, v_j) \in A$.  Then in $N$, there is a node $s$ of $S_{v_i}$ such that 
there exists a directed path $P_1$ from $s$ to $v_{j,1}$ containing a secondary arc $(t_1, t_1')$, 
and a directed path $P_2$ from $s$ to $v_{j,2}$ containing a secondary arc $(t_2, t_2')$, 
and such that
$D_{i,j}$ uses these secondary arcs (i.e. for each $h \in \{1,2\}$, 
either $(\alpha_i(x), \alpha_{i + 1}(x)) = (t_h, t_h')$ for some $x \in V(D_{i,j})$ and integer $i$, 
or $(\alpha_\last(x), \alpha_1(y)) = (t_h, t_h')$ for some $x, y \in V(D_{i,j})$).
\man{Note that $(t_1, t'_1) = (t_2, t'_2)$ is possible.}
\end{nclaim}

\begin{proof}
Let $x$ be a node of $D_{i,j}$ satisfying Claim~\ref{claim:in_svi} above.  Since $s:= \alpha_\last(x)$ is in 
the $S_{v_i}$ subtree, and that $x$ has descendants $w^i_{j,1}$ and $w^i_{j, 2}$ mapped to $v_{j,1}$ and $v_{j,2}$, there must be a path from $s$ to $v_{j, 1}$ and from $s$ to $v_{j ,2}$.
Since $s$ and $v_{j,1}$ (or $v_{j, 2}$) are incomparable in $S$, 
these paths must contain a secondary arc.  
Moreover, there must be such paths $P_1$ and $P_2$ 
and some node of $D_{i,j}$ on the $x - v_{j ,1}$ path (resp. the $x - (v_{j,2})$ path)
that uses the $(t_1, t_1')$ arc (resp. the $(t_2, t_2')$ arc).
\end{proof}

As specified in the previous claim, $(t_1, t_1') = (t_2, t_2')$ is possible.  In essence, this happens when $S_{v_i}$ is able to get to $S_{v_j}$.
In the following, let $\hat{A} \subseteq A$ be the set of arcs such that $(v_i,v_j) \in \hat{A}$ if and only if there is 
a directed path in $N$ from $r(S_{v_i})$ to $r(S_{v_j})$.
The set $A' = A \setminus \hat{A}$ will form our feedback arc set, i.e. the arcs to remove to eliminate all cycles.

\begin{nclaim}\label{claim:hat-has-nocycle}
$H' = (V, \hat{A})$ contains no directed cycle.
\end{nclaim}

\begin{proof}
Suppose instead that in $H'$, there is a cycle $C = x_1x_2\ldots x_{\l}x_1$.  
By the definition of $\hat{A}$, in $N$ there is a directed path from $r(S_{x_i})$ to 
$r(S_{x_{i + 1}})$ for every $i \in [\l - 1]$, and from $r(S_{x_{\l}})$ to $r(S_{x_1})$.
Thus $N$ contains a cycle, contradicting time-consistency.
\end{proof}

\begin{nclaim}
$|\hat{A}| \geq m - k$.
\end{nclaim}

\begin{proof}
Recall that by Claim~\ref{claim:gprime-tl}, $D$ has a transfer in $D'_{i,j}$ that is not in $D_{i,j}$,
and these together take up $m$ transfers.
Moreover by Claim~\ref{claim:gij-paths}, each $D_{i,j}$ subtree uses at least one transfer.
Since $D$ uses at at most $K = 2m + k$ transfers, there can be at most $k$ of the $D_{i,j}$ subtrees that
use more than one transfer, and hence at least $m - k$ that only use one.

By Claim~\ref{claim:gij-paths}, for each $(v_i, v_j) \in A$, there is a directed path $P_1$ in $N$ from 
$r(S_{v_i})$ to $v_{j, 1}$ and a directed path $P_2$ from $r(S_{v_i})$ to $v_{j, 2}$, such that 
$D_{i,j}$ uses the transfer arc $(t_1, t_1')$ from $P_1$ and $(t_2, t_2')$ from $P_2$.  
If $D_{i,j}$ uses one transfer, we must have $(t_1, t_1') = (t_2, t_2')$.  
This is only possible if $t_1' = t_2'$ is an ancestor of $\lca_S(v_{j,1}, v_{j,2}) = r(S_{v_j})$.
This shows that there are at least $m - k$ subtrees $D_{i,j}$, and hence arcs $(v_i, v_j)$ such that $N$ has a path from 
$r(S_{v_i})$ to $r(S_{v_j})$.
\end{proof}

We are done with the proof, since 
$A' = A \setminus \hat{A}$ is a feedback arc set of $H$ by Claim~\ref{claim:hat-has-nocycle},
and $|A'| = |A| - |\hat{A}| \leq m - (m - k) = k$.
\end{proof}

We have shown that that $H$ has a feedback arc set of size $k$
if and only $D$ is \ml{$S$-base-reconcilable} using $K = 2m + k$ transfers.
By Lemma~\ref{lem:equiv-ds-tree-unknown}, $H$ has a feedback arc set of size $k$ 
if and only if the relation graph $R(D)$ is \ml{$S$-base-consistent} using $K$ transfers.
Therefore we get the following.\\

\noindent
\textbf{Theorem \ref{thm:hard-unknown-highways}.} 
\emph{The TMSTC problem is NP-hard, even if the input relation graph $R$ has a corresponding 
least-resolved $DS$-tree that is binary.}

\end{document}